\documentclass[letterpaper,11pt]{article}
\usepackage{bm} 
\usepackage[table,xcdraw,svgnames, dvipsnames]{xcolor}
\usepackage{microtype}
\usepackage[utf8]{inputenc}
\usepackage{graphicx,fullpage,paralist}
\usepackage{caption}
\usepackage{subcaption}
\usepackage{amsmath, amssymb, amsthm}
\usepackage{comment,hyperref}
\usepackage{thmtools}
\usepackage{tikz}
\usepackage{pgfplots}
\usepackage{varwidth}
\pgfplotsset{compat=1.10}
\usepgfplotslibrary{fillbetween}
\usetikzlibrary{backgrounds}
\usetikzlibrary{patterns}
\usetikzlibrary{positioning}
\usepackage{enumitem}
\usepackage{geometry}
\usepackage{array}
\usepackage{multirow}
\usepackage{stmaryrd}
\usepackage[font=small]{caption}
\usepackage{makecell}
\usepackage{bbold}
\usepackage{cite}
\usepackage{booktabs}
\usepackage[skip=3pt, indent=20pt]{parskip}
\usepackage{setspace}
\usepackage[capitalise]{cleveref}
\usepackage{mathtools}

\newcommand{\PP}{\mathbf{P}}
\newcommand{\calR}{\mathcal{R}}
\newcommand{\calB}{\mathcal{B}}

\newcommand{\calP}{\mathcal{P}}

\newcommand{\calT}{\mathcal{T}}

\newcommand{\RR}{\mathbb{R}}
\newcommand{\NN}{\mathbb{N}}
\newcommand{\N}{\mathbb{N}}

\newcommand{\EE}{\mathbf{E}}
\newcommand{\E}{\mathbf{E}}

\newcommand{\apx}{\normalfont{\text{apx}}}
\newcommand{\aux}{\mu}

\newtheorem{theorem}{Theorem}
\newtheorem{lemma}{Lemma}
\newtheorem{corollary}{Corollary}

\newtheorem{proposition}{Proposition}
\newtheorem{definition}{Definition}

\newtheorem{claim}{Claim}

\def\prn#1{\left( #1 \right)}
\usepackage{nicefrac}
\def\f#1#2{\ensuremath{\nicefrac{#1}{#2}}}
\newcommand{\cD}{\mathcal{D}}
\newcommand{\Rp}{\Real_{+}}

\newcommand{\ACCMmax}{\normalfont{\mathcal{C}({F})}}

\newcommand{\Gmax}[1]{\normalfont{G_{{#1}}(F)} }

\newcommand{\rva}[1]{\normalfont{\text{RV}_{#1}}}
\def\set#1{\left\{ #1 \right\}}
\def\Real{\mathbb{R}}
\def\dif#1{\mathrm{d} #1}
\let\R\Real
\newcommand{\inv}{-1}
\def\bigO#1{\operatorname{O}\prn{#1}}

\def\bigOm#1{\operatorname{\Omega}\prn{#1}}

\usepackage[textsize=small,textwidth=2cm]{todonotes}

\usepackage{multirow}
\usepackage[noend]{algpseudocode}
\usepackage{algorithm,algorithmicx}
\usepackage{xcolor}

\newlength{\algofontsize}
\hypersetup{
	colorlinks,
	linkcolor={red!50!black},
	citecolor={blue!50!black},
	urlcolor={blue!80!black}
}

\begin{document}
	\algrenewcommand\algorithmicrequire{\textbf{Input:}}
	\algrenewcommand\algorithmicensure{\textbf{Output:}}
	
	\title{Posted Pricing and Competition in Large Markets
 \vspace{.5cm}
 }
	
	\author{José Correa
	\thanks{Department of Industrial Engineering, Universidad de Chile.}
    \and Vasilis Livanos
	\thanks{Center for Mathematical Modeling,
		Universidad de Chile.}
	\and Dana Pizarro
	\thanks{TBS Business School, Department of Information, Operations and Management Sciences.}
	\and 	Victor Verdugo		
	\thanks{Institute for Mathematical and Computational Engineering, Pontificia Universidad Católica de Chile}
 \thanks{Department of Industrial and Systems Engineering, Pontificia Universidad Católica de Chile}
	}
\date{\vspace{-1em}}
\maketitle

\begin{abstract}

Posted price mechanisms are prevalent in allocating goods within online marketplaces due to their simplicity and practical efficiency. We explore a fundamental scenario where buyers' valuations are independent and identically distributed, focusing specifically on the allocation of a single unit. 
Inspired by the rapid growth and scalability of modern online marketplaces, we investigate optimal performance guarantees under the assumption of a significantly large market.
We show a large market benefit when using fixed prices, improving the known guarantee of $1-1/e\approx 0.632$ to $0.712$. 
We then study the case of selling $k$ identical units, and we prove that the optimal fixed price guarantee approaches $1-1/\sqrt{2k \pi}$, which implies that the large market advantage vanishes as $k$ grows.
We use real-world auction data to test our fixed price policies in the large market regime.
Next, under the large market assumption, we show that the competition complexity for the optimal posted price mechanism is constant, and we identify precise scaling factors for the number of bidders that enable it to match benchmark performance. Remarkably, our findings break previously established worst-case impossibility results, underscoring the practical robustness and efficiency of posted pricing in large-scale marketplaces.
\end{abstract}

%
%
%
%
%
%

\thispagestyle{empty}
\newpage

\section{Introduction}

Understanding the worst-case welfare obtained by simple pricing mechanisms is a fundamental question in Economics and Computation \cite{alaei2013simple,alaei2019optimal,cohen2016invisible,feng2019optimal,hartline2015non,JJLZ21,jin2020tight,CFHOV17}. 
One of the most fundamental settings corresponds to selling a single item to $n$ buyers with valuations given by independent and identically distributed random variables. 
In this context, the simplest mechanism is to set a fixed price (a.k.a. anonymous price) for the item. 
This problem has been extensively studied in the literature, and a result of Ehsani et al. \cite{Ehsani2018} establishes that the performance of a fixed price policy when facing samples from i.i.d. random variables is at least a fraction $1-1/e$ of that of the optimal mechanism, and the bound is best possible. 
Another natural family of mechanisms considered in the literature is the one of posted price mechanisms (PPM) in which the seller offers sequentially a possibly different price to each buyer~\cite{CEGMM10,Chawla2010,CFHOV17,jin2020tight, Y11,CFHOV17,CFPV19}. 
Hill and Kertz studied the i.i.d. setting of this problem \cite{HK82}, providing a family of worst possible instances, from which Kertz \cite{K86} proved that no PPM can extract more than a fraction of roughly $0.745$ of the optimal mechanism.
Later, Correa et al. \cite{CFHOV17} proved this value is tight. 

In this work, we study optimal welfare guarantees of fixed price policies in the i.i.d. setting under an additional \emph{large market} assumption, relevant to most modern online marketplaces.
The essential difference with the classic setting is that, rather than considering the (classic) worst distribution for each possible market size $n$, we fix the distribution and then take $n$ growing large.
Furthermore, we also explore the interplay between the market size and the optimal dynamic pricing policy through the {\it competition complexity} measure, exhibiting fundamental differences with respect to the classic setting.
\subsection{Our Results and Contribution}

In what follows, we detail our contribution on tight guarantees for to the design fixed price policies in large markets and the competition complexity of the optimal dynamic pricing policy.
At the end of this section, we summarize our work's main results in Table~\ref{tab:results}.

\paragraph{Fixed Price Approximation Guarantees.} For every positive integer $n$, consider i.i.d. valuations $X_1,X_2,\ldots,X_n$ with $X_j$ distributed according to $F$ for every $j\in \{1,\ldots,n\}$, where $F$ is a distribution over the non-negative reals. Given a value $T$, consider the simple policy given by selling the item to any buyer with valuation exceeding $T$.
Then, for each distribution $F$, we are interested in understanding the limit ratio between the welfare of this simple policy, which is simply given by the probability of having an $X_i$ above $T$, $1-F^n(T)$ times the expected value of $X_i$ conditioned on it being larger than $T$,
and the welfare of the optimal mechanism, which is the expectation of the maximum valuation, and it is denoted as $M_n$. Our quantity of interest is then
\begin{equation}
\label{eq:ratio2}
\apx(F)=\liminf_{n\to \infty}\sup_{T\in \RR_+}\frac{1-F^n(T)}{\EE(M_n)}\left(T+\frac{1}{1-F(T)}\int_{T}^{\infty}(1-F(s))\mathrm ds\right),
\end{equation}
and we call this value {\it welfare guarantee}.
Our first main result shows that $0.712$ is a tight lower bound for $\apx(F)$ when the distribution satisfies the {\it extreme value} condition.
This condition is, essentially, the equivalent of a central limit theorem for the maximum of an i.i.d. sample, and it is the cornerstone of the extreme value theory. The $0.712$ value is substantially better than the known bound of $1-1/e$ by Ehsani et al. \cite{Ehsani2018} and thus represents a significant advantage for the large markets setting. We remark that we are mainly interested in the case of distributions $F$ with unbounded support, since one can show that $\apx(F)=1$ when $F$ is of bounded support.

A natural and practically relevant extension of the problem of selling a single item is to consider the setting in which we can sell up to $k$ items. We call this problem $k$-unit problem and we study whether the large market advantage continues to be significant beyond the single item case. 
To this end, we provide a lower bound for the approximation factor achievable by the best fixed price policy, again under the extreme value condition. More specifically, for each value of $k$, the approximation factor is bounded by a (computationally) simple optimization problem.
In particular, the bound presented when $k=1$ follows from obtaining the exact solution of the optimization problem.
The performance obtained by our characterization represents an advantage for the $k$-unit problem. 

However, we show that the advantage vanishes as $k\to \infty$. Indeed, we prove that for each integer $k$, the approximation factor is more than $1-1/\sqrt{2k\pi}$, but there exists $F$ such that this lower bound is asymptotically tight in $k$. 
This tightness, together with the $1-1/\sqrt{2k\pi}$ lower bound on the fixed-price approximation ratio of the $k$-unit problem (see, e.g., \cite{DFK16,arnostma2020})
implies that the large market advantage vanishes as $k\to \infty$; for an illustration, Figure~\ref{fig:minmaxk}
depicts the bound obtained by our optimization problem and compares it with $1-1/\sqrt{2k\pi}$.
\begin{figure}[t]
\centering
\includegraphics[scale=0.4]{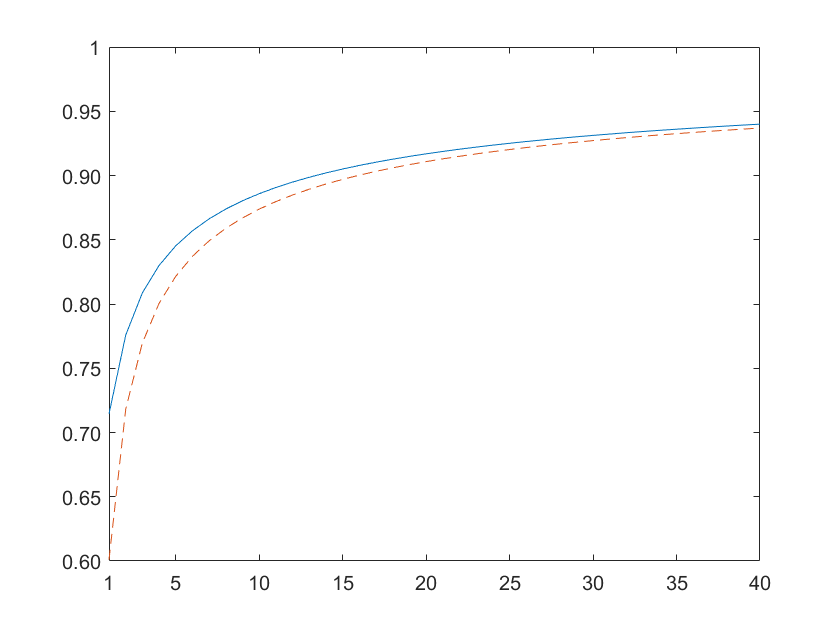}
\caption{\small Our optimal welfare guarantee over $k$ (continuous line) vs. $1-1/\sqrt{2k\pi}$ (dashed line).}
\label{fig:minmaxk}
\end{figure}
We note that as a direct corollary, when $F$ satisfies the extreme value condition and for large markets, we can derive 
the {\it adaptivity gap}; the worst case ratio between the value obtained by the optimal fixed price policy and that of the optimal PPM of Kennedy and Kertz. This value is at most 1.105.

Finally, we give sufficient conditions under which our main result for the $k$-unit problem can be translated into a revenue guarantee for a fixed price policy. More specifically, while by standard arguments this works as long as the distribution of the virtual values of $F$, call it $G$, satisfies the extreme value condition, we investigate the following question: When $F$ satisfies the extreme value condition, can we guarantee that the distribution of the virtual valuation $G$ also does? And, if this is the case, do $G$ and $F$ fall in the same extreme value family? We answer these questions positively under some mild assumptions.
The formal statements of our results in this line can be found in Section \ref{sec:single-selection-proof}.

In \Cref{sec:numerical}, we use real data from an eBay auction to explore the real-world deployment of our findings, and we calculate how much the seller can guarantee if they set a fixed price for a Cartier watch instead of using the optimal mechanism to sell it. 

\paragraph{Competition Complexity.} Due to the inherent complexity of optimal pricing mechanisms, researchers have turned their attention to pricing and auction design of more straightforward, albeit suboptimal, mechanisms.
In this line, the competition complexity question asks how many more bidders are required for the simpler mechanism to match the performance of the optimal one.
The seminal work by Bulow and Klemperer~\cite{bulow-klemperer} states that the simple second-price auction with $n+1$ bidders achieves a revenue that is at least as good as that of Myerson's optimal auction with $n$ bidders, so long as the bidders' valuations are drawn i.i.d. from a regular distribution. 
In recent years, this approach has been further explored in the context of pricing and prophet inequalities~\cite{beyhaghi2019optimal,bulow-klemperer,eden2016competition,feldman201899,liu2018competition,correa-competition-complexity,brustle2024competition,ezra2024competition1,ezra2024competition2}.
The message in the competition paradigm is simple; rather than designing a better pricing mechanism, it is better to run a simple one and attract more bidders.

Unfortunately, Brustle et al. \cite{correa-competition-complexity,brustle2024competition} show a strong impossibility: The competition complexity can be unbounded in the worst-case when the distribution can depend on $n$. In contrast, we ask whether the same impossibility holds under the large market assumption. 
This motivates the following question that we address in this work: \emph{Under the extreme value condition, and in the large market regime, can we break the impossibility for the competition complexity of dynamic pricing?}
We answer this question in the affirmative and provide an exact characterization of the competition complexity for large markets.
Namely, in the large market regime, we provide tight (constant) factors for the scaled number of bidders needed to beat the welfare achievable by an optimal mechanism.
The formal statement of our result can be found in \Cref{sec:comp-complexity}.

\begin{table}[!h]
\begin{center}
\begin{tabular}{| c | c | c | c | }
\hline
 & \multicolumn{3}{ c |}{Extreme type of the distribution $F$}  \\ 
\cline{2-4}
& Fréchet & Gumbel & Reversed Weibull  \\
\hline
\hline

FP/OPT, WG, $k=1$ & $\approx 0.712$ & 1 &1\\ 
& Thm. \ref{thm:apx:multi}\ref{multiple-frechet}& Thm. \ref{thm:apx:multi}\ref{multiple-Gumbel}& Thm. \ref{thm:apx:multi}\ref{multiple-Gumbel} \\
 \hline
FP/OPT, WG, any $k$ & $\approx 1-1/\sqrt{2 \pi k}$ & 1  & 1  \\ 
& Thm. \ref{thm:apx:multi}\ref{multiple-frechet}& Thm. \ref{thm:apx:multi}\ref{multiple-Gumbel}& Thm. \ref{thm:apx:multi}\ref{multiple-Gumbel} \\ \hline
DP/OPT, CC, $k=1$ & $1 \leq \ACCMmax \leq 1.781$ & $\ACCMmax \approx 1.781$ & $  1.781 \leq \ACCMmax \leq  e $  \\
& Thm. \ref{thm:comp-complexity-evt}\ref{thm:CC-frechet}, Cor. \ref{cor:comp-complexity-families}\ref{cor:cc-frechet} & Thm. \ref{thm:comp-complexity-evt}\ref{thm:CC-gumbel}, Cor. \ref{cor:comp-complexity-families}\ref{cor:cc-gumbel} & Thm. \ref{thm:comp-complexity-evt}\ref{thm:CC-weibull}, Cor. \ref{cor:comp-complexity-families}\ref{cor:cc-weibull} \\
\hline
\end{tabular}
\vspace{0.1cm}
\caption{Main results under the large market assumption for each extreme value distribution family. The table reports the welfare guarantees (WG) of the fixed price (FP) relative to the optimal mechanism (OPT), both in the single-selection case ($k = 1$) and in the general multi-selection case (any $k$). 
It also presents the competition complexity $\ACCMmax$ for a distribution $F$ (CC) of the optimal dynamic pricing policy (DP) relative to the optimal mechanism (OPT).}
\label{tab:results}
\end{center}
\end{table}

\subsection{Related Work} \label{sec:litrev}
The study of prophet inequalities was initiated by Gilbert and Mosteller \cite{GM66}, and since then, it has been studied under different settings and model variants \cite{HK82,K86,KW12,SM02,SC83}.
In the last two decades, the problem gained particular attention due to its close connection with online mechanisms.
This relation was first studied by Hajiaghayi et al. \cite{HKS07}, who showed that prophet inequalities can be interpreted as posted price mechanisms for online selection problems. Later, Chawla et al. \cite{Chawla2010} proved that any prophet inequality can be turned into a posted price mechanism with the same approximation guarantee. The reverse direction was proven by Correa et al. \cite{CFPV19} and thus the guarantees for optimal stopping problems are equivalent to the problem of designing posted price mechanisms. 
Furthermore, in the i.i.d. setting, fixed threshold stopping rules become equivalent to fixed price policies; this equivalence allows one to transfer directly some results in i.i.d. prophet inequalities to the corresponding pricing problem.
We remark that Ehsani et al. \cite{Ehsani2018} prove that the bound of $1-1/e$ is best possible for the i.i.d. prophet inequality with fixed threshold policies by constructing a sequence of distributions depending on the number of buyers $n$, while in our large market setting, the distribution $F$ is fixed.
In addition to the fixed price policies (FP) and the optimal mechanism (OPT), the most common mechanisms considered are general PPMs and the second price auction with a fixed reserve price. We note here that there is a vast amount of recent literature comparing these four mechanisms under different model variants \cite{alaei2019optimal, DFK16, HR09,JJLZ21, JLQTX19, jin2020tight}. 

Regarding the $k$-unit problem, the worst-case guarantee in the i.i.d. setting has been proven to be at least $1-k^ke^{-k}/k! \approx 1-1/\sqrt{2\pi k}$; see, e.g.,~\cite{CEGMM10,Y11,DFK16,BGPS20,arnostma2020}. 
Recently, there have been several results regarding the worst-case guarantee of the optimal dynamic programming policy for the $k$-unit prophet inequality in the i.i.d. setting. 
Jiang et al. \cite{jiang2022tightness} introduced an optimization framework to characterize worst-case guarantees for prophet inequality problems, including the i.i.d. $k$-unit setting for a fixed $n$.
Brustle et al. \cite{brustle2024splitting} characterize the worst-case guarantee for the $k$-unit prophet inequality in the i.i.d. setting, and fixed $n$, via an infinite-dimensional linear program in the space of quantiles; they also provide a closed-form nonlinear system of differential equations that gives provable lower bounds on the worst-case guarantee when $n$ is large enough. 
Using different techniques, Molina et al. \cite{molina2025prophet} also derived a similar nonlinear system.

Under the large market regime, the seminal work by Kennedy and Kertz \cite{KK91} compares the optimal PPM with the optimal mechanism in the single item setting, and proves a welfare guarantee of 0.776. 
We highlight two recent works under the large market regime that complement ours.
Livanos and Mehta \cite{livanos2024minimization} characterize the approximation guarantee for the optimal policy for the cost minimization variant of the prophet inequality through the lens of extreme value theory. On the other hand, Abdallah and Reed \cite{abdallah2025regime} characterize the optimal dynamic pricing policy in continuous time, where the number of arrivals is random, in the single-item case. 
We remark that our results and the two recent works previously mentioned also contribute to the literature on the study of approximation guarantees for models beyond worst-case; see, e.g., \cite{feng2024beyond,goldenshluger2022optimal,giannakopoulos2021optimal,cai2015extreme,livanos2024minimizationSODA,arsenis2021constrained}.

\section{Preliminaries}\label{sec:prelim}

We recall that $F$ is a distribution if it is a right-continuous and non-decreasing function, with limit equal to zero in $-\infty$ and equal to one in $+\infty$. 
We consider $F$ to be absolutely continuous in $\RR$, and we denote its density by $f$ or $F'$, depending on the context.
In general, $F$ is not invertible, but we work with its generalized inverse, given by $F^{-1}(y)=\inf\{t\in \RR:F(t)\ge y\}$. 
We denote by $\omega_0(F)=\inf\{t\in \RR:F(t)>0\}$ and $\omega_1(F)=\sup\{t\in \RR:F(t)<1\}$, and we call the interval $(\omega_0(F),\omega_1(F))$ the support of $F$.
Given an i.i.d. sample $X_1,\ldots,X_n$ distributed according to $F$, we denote by $M_n=\max_{j\in \{1,\ldots,n\}}X_j$.


One of the main goals in the extreme value theory is to understand the limit behavior of the sequence $\{M_n\}_{n\in \NN}$.
As the central limit theorem characterizes the convergence in distribution of the average of random variables to a normal distribution, a similar result can be obtained for the sequence of maxima $\{M_n\}_{n\in \NN}$, but this time there are three possible limit situations. 
One of the possible limits is the Gumbel distribution $\Lambda(t)=\exp(-e^{-t})$; we call these distributions the Gumbel family.
Given $\alpha>0$, the second possible limit is the Fr\'echet distribution of parameter $\alpha$, defined by $\Phi_{\alpha}(t)=\exp(-t^{-\alpha})$ if $t\ge 0$, and zero otherwise; we call these distributions the Fr\'echet family.
Finally, given $\alpha>0$, the third possibility is the reversed Weibull distribution of parameter $\alpha$, defined by $\Psi_{\alpha}(t)=\exp(-(-t)^{\alpha})$ if $t\le 0$, and one otherwise; we call these distributions the reversed Weibull family. 
In Appendix \ref{app:EVT-prelim} we formally state the extreme value theorem, a result due independently to Gnedenko \cite{gnedenko1943distribution} and Fisher \& Tippett~\cite{fisher28}, and some other classic results that we use in our technical analysis.
In the following, we say that a distribution $F$ satisfies the {\it extreme value condition} if there exist sequences $\{a_n\}_{n\in \NN}$, that we call the {\it scaling sequence}, and $\{b_{n}\}_{n\in \NN}$, that we call the {\it shifting sequence}, such that $(M_n-b_n)/a_n$ converges in distribution to a random variable; see Table \ref{table:table1} for a summary and examples of distributions satisfying the extreme value condition.
We remark that there are examples of continuous distributions not satisfying this extreme value condition, e.g., $F(x)=\exp(-x-\sin(x))$.

It can be shown that for every distribution $F$ with extreme type in the reversed Weibull family
we have $\omega_1(F)<\infty$~\cite[Proposition 1.13, p. 59]{libroextremo}.
When $F$ has extreme type Fr\'echet, we have $\omega_1(F)=\infty$~\cite[Proposition 1.11, p. 54]{libroextremo}.
For the distributions with extreme type Gumbel the picture is not so clear since $\omega_1(F)$ is neither finite nor unbounded in general.
\begin{table}[t]
	\begin{tabular}{|c|c|c|c|}
	\hline
		Extreme Type & Parameter & Limit Distribution  & Example\\
		\hline \hline
		Gumbel & None &$\exp(-e^{-t})$ & Exponential distribution\\
		\hline
		Fr\'echet & $\alpha\in (0,\infty)$ & $\exp(-t^{-\alpha})\cdot \mathsf{1}_{[0,\infty)}$ & Pareto distribution \\
		\hline
		Reversed Weibull & $\alpha\in (0,\infty)$ & $\exp(-(-t)^{\alpha})\cdot \mathsf{1}_{(-\infty,0)}+\mathsf{1}_{[0,\infty)}$ & Uniform distribution\\
		\hline
	\end{tabular}
\vspace{0.1cm}
	\caption{Summary of the three possible extreme value distributions. The Fr\'echet family and the Reversed Weibull family are associated to a parameter $\alpha\in (0,\infty)$. Recall that for $\alpha>0$, the Pareto distribution of parameter $\alpha$ is given by $1-t^{-\alpha}$ for $t\ge 1$ and zero otherwise.}
		\label{table:table1}
\end{table}

\section{Welfare and Revenue Guarantees in Large Markets}
\label{sec:single-selection-proof}

We say that a policy for the $k$-unit problem with i.i.d. valuations $X_1,X_2,\ldots,X_n$ is a {\it fixed price policy} if there exists a value $T$ such that we sell an item to the first $\min\{k,|Q|\}$ buyers with valuation attaining a value larger than $T$, where $Q$ is the subset of valuations attaining a value larger than $T$. Under this policy, the items are allocated to all buyers in $Q$ or to a random subset of size $k$.
Consider the random variable $\calR^n_{k,T}$ equal to the summation of the first $\min\{k,|Q|\}$ valuations attaining a value larger than $T$.
In particular, 
this value is completely determined by the number of buyers $n$, the distribution $F$ and the price $T$.
We are interested in understanding the value
\begin{equation}
	\apx_k(F)=\liminf_{n\to \infty}\sup_{T \in \RR_+}\frac{\EE(\calR^n_{k,T})}{\sum_{j=1}^k \EE(M_n^{j})},\label{eqn:comp:ratio:k:def}
\end{equation}
where $M_n^1\ge M_n^2\ge \cdots\ge M_n^n$ are the order statistics of a sample of size $n$ according to $F$.
We observe that when $k=1$ the value $\apx_k(F)$ corresponds to the value $\apx(F)$ in \eqref{eq:ratio2}.

{
We start by observing that when $\omega_1(F)<\infty$, which is the case for the reversed Weibull family, we have $\apx_k(F)=1$ for every positive integer $k$. 
When the support of $F$ is upper bounded by $\omega_1(F)<\infty$, we have $\EE(M_n^j)\le \omega_1(F)$ for every $j\in \{1,\ldots,k\}$. 
For every $\varepsilon>0$ consider $T_{\varepsilon}=(1-\varepsilon)\cdot \omega_1(F)$.
Then, by the expression in \eqref{eqn:comp:ratio2:k}
we have that $\apx_k(F)$ can be lower bounded as
$\apx_k(F)\ge (1-\varepsilon)\cdot \omega_1(F)\cdot \liminf_{n\to \infty}\sum_{j=1}^{k} \PP(M_n^{j}>T_{\varepsilon})/(k\cdot \omega_1(F))=1-\varepsilon$,
and we conclude that $\apx_k(F)=1$.}
\subsection{Welfare Guarantees}

Now we formally present our main results for welfare guarantees in the $k$-unit problem. 

\begin{theorem} 
	\label{thm:apx:multi}
	Let F be a distribution over the non-negative reals that satisfies the extreme value condition. Then, the following holds.
\begin{enumerate}[label=\normalfont(\alph*)]
	\item When $F$ has an extreme type Fr\'echet of parameter $\alpha$, for every positive integer $k$ we have that $apx_k(F) \ge \varphi_k(\alpha),$ where  $\varphi_k:(1,\infty)\to \RR_+$ is given by
\begin{equation}
	\label{eq:phi-function}
	\varphi_k(\alpha)=\frac{\Gamma(k)}{\Gamma(k+1-1/\alpha)} \max_{x\in (0, \infty)} x \exp(-x^{-\alpha}) \sum_{j=1}^k \sum_{s=j}^{\infty} \frac{x^{-s \alpha}}{s!}.
\end{equation}	

In particular, we have $\apx_1(F)\ge 0.712$ and $\apx_k(F)\ge 1-1/\sqrt{2\pi k}$ for every $k$.
. \label{multiple-frechet}
	\item When $F$ has extreme type in the Gumbel or reversed Weibull families, we have that $\apx_k(F)=1$ for every positive integer $k$.\label{multiple-Gumbel}
\end{enumerate}
\end{theorem}


\begin{theorem}
\label{thm:tightness}
Let $F$ be the Pareto distribution with parameter $\alpha=2$. Then, for every $\varepsilon>0$ there exists a positive integer $k_{\varepsilon}$ such that for every $k\ge k_{\varepsilon}$ it holds that $\apx_{k}(F) \leq  1-(1-\varepsilon)/\sqrt{2\pi k}$.
\end{theorem}


We defer the proof of Theorem \ref{thm:apx:multi} to Appendix \ref{sec:proof_main_thm}, and the proof of Theorem \ref{thm:tightness} to Appendix \ref{app:tightness}. Observe that by Theorem~\ref{thm:apx:multi} we have that for each integer $k$ the approximation factor is more than $1-1/\sqrt{2k\pi}$ under the large market assumption. Moreover, by Theorem~\ref{thm:tightness} this lower bound is in fact asymptotically tight in $k$ for the distributions with extreme type Fr\'echet of parameter $\alpha=2$. This tightness, together with the recent result of Duetting et al. \cite{DFK16} establishing that the approximation ratio of the $k$-unit problem without the large market assumption is almost $1-1/\sqrt{2k\pi}$, allows us to obtain the result that the large market advantage vanishes as $k\to \infty$.

Despite the tightness result established in Theorem~\ref{thm:tightness}, for small values of $k$ this bound is in fact substantially better.
Consider a distribution $F$ with extreme type Fr\'echet of parameter $\alpha\in (1,\infty)$.
By Theorem \ref{thm:apx:multi}\ref{multiple-frechet}, when $k=1$ it holds that 
\begin{equation}\label{eq:opt-U-for-1}
	\varphi_1(\alpha)=\frac{1}{\Gamma(2-1/\alpha)}\sup_{x\in (0,\infty)}x\Big(1-\exp(-x^{-\alpha})\Big),
\end{equation}
for every $\alpha\in (1,\infty)$.
The optimum for the above optimization problem as a function of $\alpha $ is attained at the smallest real non-negative solution $U^*(\alpha)$ of the first order condition $U^{\alpha}+\alpha=U^{\alpha}\exp(U^{-\alpha})$, which is given by
\[U^*(\alpha)=\left(-\frac{1}{\alpha}\left(\alpha W_{-1}\left(-\frac{1}{\alpha}e^{-1/\alpha}\right)+1\right)\right)^{-1/\alpha},\]
where $W_{-1}$ is the negative branch of the Lambert function. 
Therefore, we have
\[\varphi_1(\alpha)=\frac{\alpha}{\Gamma(2-1/\alpha)}\cdot \frac{U^*(\alpha)}{U^*(\alpha)^{\alpha}+\alpha}.\]
The minimum value is at least $0.712$ and it is attained at $\alpha^*\approx 1.656$.
Note that when $\alpha$ approaches to zero or $\infty$, the function $\varphi_1$ goes to one and thus the unique minimizer is given by $\alpha^*\approx 1.656$. 
We highlight here that, even though Theorem~\ref{thm:apx:multi} implies that $\apx(F)$ is at least $\varphi_1(\alpha^*) \approx 0.712$ when $F$ has extreme type Fréchet, this bound is in fact reached by the Pareto distribution with parameter $\alpha^*$ and therefore this bound is tight.

Given our closed expression for the function $\varphi_1$, we can compare it with the closed expression obtained by Kennedy and Kertz for the welfare guarantees of the optimal dynamic policy \cite{KK91}. 
 Given a distribution $F$, for every positive integer $n$ let $v_n=\sup\{\EE(X_{\tau}):\tau\in \calT_n\}$ and consider the stopping time given by $\tau_n=\min\{k\in \{1,\ldots,n\}:X_k>v_{n-k}\}$.
 In particular, $v_n=\EE(X_{\tau_n})$ for every positive integer $n$.
 The following summarizes the result of Kennedy and Kertz \cite{KK91} for the optimal dynamic policy:
When $F$ is a distribution in the Fr\'echet family,
there exists $\nu:(1,\infty)\to (0,1)$ such that $\lim_{n\to \infty}v_n/\EE(M_n)=\nu(\alpha)$ when $F$ has an extreme type Fr\'echet of parameter $\alpha$.
Furthermore, $\lim_{\alpha\to \infty}\nu(\alpha)=\lim_{\alpha\to 1}\nu(\alpha)=1$ and $\nu(\alpha)\ge 0.776$ for every $\alpha\in (1,\infty)$.
The function $\nu$ is given by 
 \[\nu(\alpha)=\frac{1}{\Gamma(2-1/\alpha)}\left(1-\frac{1}{\alpha}\right)^{1-\frac{1}{\alpha}},\]
 and we have $\varphi_1(\alpha)\le \nu(\alpha)$ for every $\alpha\in (1,\infty)$.
 
 Kennedy and Kertz show that the asymptotic approximation obtained by their policy when the distribution has an extreme type in the Gumbel and reversed Weibull family is equal to one.
 Our Theorem~\ref{thm:apx:multi}\ref{multiple-Gumbel} shows that for both such families we can attain this value by using just fixed price policies.
The {\it adaptivity gap} is equal to the ratio between the optimal welfare guarantee obtained by fixed price policies and the value obtained by the policy of Kennedy and Kertz.
 As a corollary of our result for $k=1$, we obtain an upper bound on the adaptivity gap for the case of distributions with extreme value.
 For the family of distributions over the non-negative reals and satisfying the extreme value condition we have that the adaptivity gap is at most 
$\max_{\alpha\in (1,\infty)}\nu(\alpha)/\varphi_1(\alpha)\approx 1.105$
 and is attained at $\alpha\approx 1.493$.

\subsection{Revenue Guarantees}\label{sec:stability}

In this section we describe how our welfare guarantees can be used to obtain guarantees for fixed price mechanisms considering the revenue objective.
The {\it virtual valuation} associated to a distribution $F$ is given by $\phi_F(t)=t-(1-F(t))/f(t)$, where $f$ is the density function of $F$.
When $v$ is distributed according to $F$, we denote by $F_{\phi}$ the distribution of $\phi_F(v)$ and by $F_{\phi}^+$ the distribution of $\phi^+_F(v)=\max\{0,\phi_F(v)\}$.

In particular, we are interested in understanding the limit ratio defined in a similar way as the welfare guarantee \eqref{eq:ratio2}, namely {\it revenue guarantee}, where the ratio is now taken between the revenue of a fixed price policy, and the revenue obtained by Myerson's optimal mechanism.
Using Theorem \ref{thm:apx:multi} we can apply the existing reductions in the literature~\cite{Chawla2010,CFPV19,HKS07} to translate our welfare guarantees to revenue guarantees, as long as $F_{\phi}^+$ satisfies the extreme value condition.
If  $F_{\phi}^+$ has extreme value Fr\'echet, the revenue guarantee of a fixed price policy for the $k$-unit problem is bounded by the maximization problem \eqref{eq:phi-function} which is more than $1-1/\sqrt{2k\pi}$, and asymptotically tight in $k$. When $k=1$, we further have the revenue guarantee of roughly 0.712. 
When $F_{\phi}^+$ is in the Gumbel or reversed Weibull families, we have that fixed price policies can recover the same revenue of that of the optimal mechanism for the $k$-unit problem, for every $k$.

When $F$ is in the Gumbel family, it admits a functional representation by the so called {\it Von Mises family}, and we say that it can be {\it smoothly represented} when the Von Mises representation satisfies a mild smoothness condition; see Appendix \ref{app:EVT-prelim} for the formal definition.
A distribution $F$ with extreme type Fr\'echet of parameter $\alpha$ satisfies the {\it asymptotic regularity condition} if
$\lim_{t\to \infty}(1-F(t))/(tf(t))=1/\alpha$. 
This holds, for example, every time that $f$ is non-increasing; see, e.g., \cite[Proposition 1.15]{libroextremo}.
We show that if a distribution $F$ satisfies any of these two conditions, the distribution $F_{\phi}^+$ has an extreme type as well, and it belongs to the same family. The proof can be found in Appendix~\ref{app:VVEV}. 

\begin{theorem}
	\label{thm:stability}
	Let $F$ be a distribution satisfying the extreme value condition. 
	The following holds:
	\begin{enumerate}[label=\normalfont(\alph*)] 
		\item If $F$ has extreme type in the Fr\'echet family with parameter $\alpha$ and it satisfies the asymptotic regularity condition, then $F_{\phi}^+$ has extreme type in the Fr\'echet family with the same parameter.
		\item If $F$ has extreme type Gumbel and it can be smoothly represented, then $F_{\phi}^+$ has extreme type Gumbel.
	\end{enumerate}
\end{theorem}

As it was mentioned above, the well established connection between posted pricing mechanisms and prophet inequalities \cite{Chawla2010,CFPV19,HKS07}, together with Theorem~\ref{thm:stability}, allows us to apply Theorem~\ref{thm:apx:multi} to obtain optimal revenue guarantees of fixed price policies in the large market setting. In what follows we state this result, whose proof follows directly from the above-mentioned results.

\begin{corollary} 
	Let $F$ be a distribution satisfying the extreme value condition. 
	The following holds:
\begin{enumerate}[label=\normalfont(\alph*)]
	\item When $F$ has an extreme type Fr\'echet of parameter $\alpha$ and it satisfies the asymptotic regularity condition, we have that the revenue guarantee for the $k$-unit problem is at least  $1-1/\sqrt{2\pi k}$, and this bound is asymptotically tight in $k$.
	\item When $F$ has extreme type in the Gumbel family and it can be smoothly represented or in the reversed Weibull family, we have that the revenue guarantee is $1$ for every positive integer $k$.
\end{enumerate}
\end{corollary}

\section{Competition Complexity Guarantees in Large Markets}\label{sec:comp-complexity}

In what follows, we denote by $\Gmax{n}$ the optimal dynamic programming policy value for an i.i.d. sample of size $n$ distributed according to $F$. The sequence $\{G_n(F)\}_{n\in \NN}$ satisfies the following recurrence: $G_0(F)=0$, $G_1(F)=\EE(X)$, and $G_{n+1}(F)=\EE(\max(X,G_n(F)))$, where $X$ is distributed according to $F$. The large market {\it competition complexity} of a distribution $F$ is
\begin{equation}
\ACCMmax= \liminf_{n\to \infty}\inf_{m\ge n}\set{m/n:\Gmax{m} \geq \EE(\textstyle\max_{i\in [n]} X_i)}.\label{eq:comp-comp}
\end{equation}
By studying the optimal policy sequence $\{G_n(F)\}_{n\in \NN}$ in the large market regime, we fully characterize the competition complexity \eqref{eq:comp-comp} of distributions satisfying the extreme value condition using a single function that depends only on the parameter $\alpha$ of the three families of extreme value distributions. To facilitate our analysis, it is easier to re-parametrize the distributions based on a single value $\gamma \in \R$ as follows: (i) for distributions in the Fr\'{e}chet family, we let $\gamma = 1/\alpha$, (ii) for distributions in the Gumbel family, we let $\gamma = 0$ and (iii) for distributions in the Reverse Weibull family, we let $\gamma = - 1/\alpha$. In this way, we can analyze the competition complexity solely based on $\gamma$, as each distribution family corresponds to a different domain for $\gamma$, namely positive, zero and negative, respectively.

In \Cref{sec:analysis-competition} we summarize the technical details about this parameterization and provide a proof of the following theorem.
\begin{theorem}\label{thm:comp-complexity-evt}
Let $F$ be a distribution over the non-negative reals that satisfies the extreme value condition. The large market competition complexity of $F$ is $\ACCMmax = (1-\gamma) \prn{\Gamma(1-\gamma)}^{1/\gamma}$. In particular,
\begin{enumerate}[label=\normalfont(\alph*)] 
    \item If $F$ is in the Fr\'{e}chet family with parameter $\alpha > 0$, then $\ACCMmax = \prn{1-\frac{1}{\alpha}} \prn{\Gamma\prn{1-\frac{1}{\alpha}}}^\alpha$.\label{thm:CC-frechet}
    \item If $F$ is in the Gumbel family, then $\ACCMmax = e^{\gamma^\star} \approx 1.781$, where $\gamma^\star \approx 0.577$ is the Euler-Mascheroni constant.\label{thm:CC-gumbel}
    \item If $F$ is in the Reverse Weibull family with parameter $\alpha > 0$, then $\ACCMmax = \prn{1+\frac{1}{\alpha}} \prn{\Gamma\prn{1+\frac{1}{\alpha}}}^\alpha$.\label{thm:CC-weibull}
\end{enumerate}
\end{theorem}
As a corollary of Theorem~\ref{thm:comp-complexity-evt}, we can provide upper and lower bounds on the large market competition complexity for every family of distributions satisfying the extreme value condition.
\begin{corollary}\label{cor:comp-complexity-families}
For every distribution $F$, the following holds:
\begin{enumerate}[label=\normalfont(\alph*)] 
    \item If $F$ is in the Fr\'{e}chet family, then $1 \leq \ACCMmax \leq e^{\gamma^\star} \approx 1.781$, where $\gamma^\star \approx 0.577$ is the Euler-Mascheroni constant.
    \label{cor:cc-frechet}
    \item If $F$ is in the Gumbel family, then $\ACCMmax= e^{\gamma^\star} \approx 1.781$.
    \label{cor:cc-gumbel}
    \item If $F$ is in the Reverse Weibull family, then $e^{\gamma^\star} \approx 1.781 \leq \ACCMmax \leq e \approx 2.718$.
    \label{cor:cc-weibull}
\end{enumerate}
Furthermore, all these bounds are asymptotically tight for $\gamma$ in the corresponding regime.
\end{corollary}

An interesting fact about Theorem~\ref{thm:comp-complexity-evt} is that the Fr\'{e}chet family has smaller competition complexity than the Gumbel family, which, in turn, has smaller competition complexity than the Weibull family. This is in stark contrast with the welfare (and revenue) approximation guarantees in Theorem \ref{thm:apx:multi}, where the ratio comparison is the opposite; it is worse for the Fr\'{e}chet family.

\section{Case Study Using eBay Auctions Data}\label{sec:numerical}

Although Theorem~\ref{thm:apx:multi} establishes a constant lower bound for the welfare guarantee under the assumption of a large market, when the distribution $F$ has an extreme type Fréchet, the theorem also enables us to compute a welfare guarantee that depends on the distribution's shape parameter $\alpha$. This dependency provides additional flexibility, making it feasible in practical applications to compute a welfare guarantee, provided that sufficient data on the buyers' valuation is available. In this section, we use real data from an eBay auction to calculate how much the seller can guarantee if they set a fixed price for a Cartier watch instead of using the optimal mechanism to sell it. In the following, we briefly introduce the main features of an eBay auction.

Auctions on the eBay platform use a fixed-time, semi-sealed-bid modified second-price format, very similar to an English auction format. In eBay auctions, the highest bidder at the end wins, but the final price is determined by the second-highest bid. Typically, this price is the second-highest bid plus one bid increment, though it can be less if the winner’s bid is not a full increment above the second-highest bid, or more if necessary to meet a reserve price. An important feature of this format is that, throughout the auction, eBay hides the full amount of the highest bid so far, displaying only a {\it current bid}, which represents the price if no further bids are placed before the auction ends. New bids must be at least one increment above that price.
If you bid less than or the same as the current highest bid, the current bid is updated to one increment above your bid. If you bid more, you assume the lead and the current price is recalculated at one increment above the previously hidden maximum bid.

Due to these specific characteristics of the auction format used by eBay, the data generated from these auctions provides insights into how buyers assess the value of items and their willingness to pay. Consequently, in this section, we will leverage this information to do an analysis of our main result in Section~\ref{sec:single-selection-proof}. We will break down the analysis into three parts. First, we will provide a detailed explanation of the data we use. This will involve cleaning and organizing the data obtained from the eBay auction to ensure its suitability for analysis. Next, we will focus on estimating the underlying distribution of the buyers' valuations and its extreme value type. Finally, we will apply the estimates derived in the previous step to determine the seller's welfare guarantee in this particular case. Specifically, we will compare the ratio between the expected welfare the seller can secure by setting a fixed price for the watch to the one obtained by using the optimal auction mechanism. 

\paragraph{Data Description.}
For our analysis, we utilized the eBay database accessible at this link: https://www.modelingonlineauctions.com/datasets. 
Specifically, we focused on data from 7-day auctions of Cartier watches. The dataset contains information on 97 auctions, comprising a total of 1348 recorded bids. 
To simplify our analysis, we will assume that all the auctions feature substitute watches, meaning that each item can be considered interchangeable with the others. As a result, the dataset can be treated as if it represents a single auction. In this hypothetical auction, all 1348 bids from the various auctions are aggregated, allowing us to analyze the bidding behavior as if it occurred in one event. This assumption helps streamline the analysis by eliminating potential differences between individual auctions and focusing on overall bidding patterns.

The primary aim of the use of this data is to derive insights into consumer valuations. To achieve this, we will examine two key variables for each bid: the bid amount and the bidder's ID.
Given the nature of the auction, where bidders are allowed to place multiple bids, it is crucial to address the issue of bid multiplicity in our analysis.
To address this, we have refined the dataset by selecting only the highest bid placed by each bidder ID. This approach assumes that the highest bid from a bidder represents their valuation or maximum willingness to pay for the item. This leaves us with 509 total bids, one per bidder. In Figure~\ref{fig:density_bids}, we illustrate the distribution of the bids using a histogram of relative frequencies. The bids are organized into bins, each with a width of 200, to provide a clearer understanding of their distribution.

Our plan is to fit an extreme value distribution to the data and then, given the estimated parameter $\hat{\alpha}$, compute the threshold $T_n$ promised by Theorem~\ref{thm:apx:multi}, for $n = 509$. Afterwards, we will compute the competitive ratio of this single threshold algorithm on the bidders' values from the data and compare it to the guarantee we provide in Theorem~\ref{thm:apx:multi}.

\begin{figure}[t]
    \centering
        \begin{subfigure}[t]{0.45\textwidth}
            \centering
            \includegraphics[width=\textwidth]{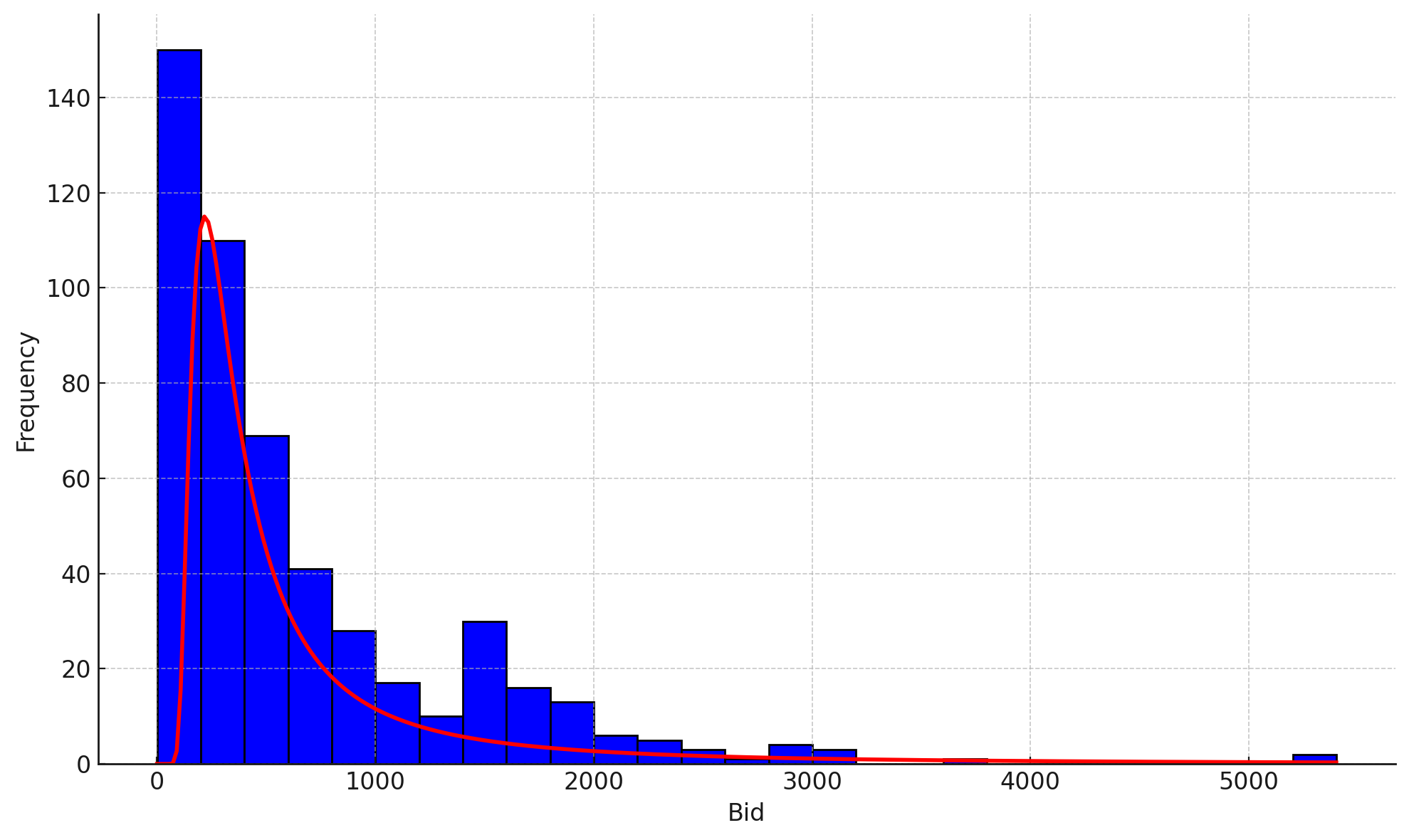}
            \caption{Frequency histogram of real bidder data.}
            \label{fig:density_bids}
        \end{subfigure}
     \hfill
     \begin{subfigure}[t]{0.45\textwidth}
         \centering
         \includegraphics[width=\textwidth]{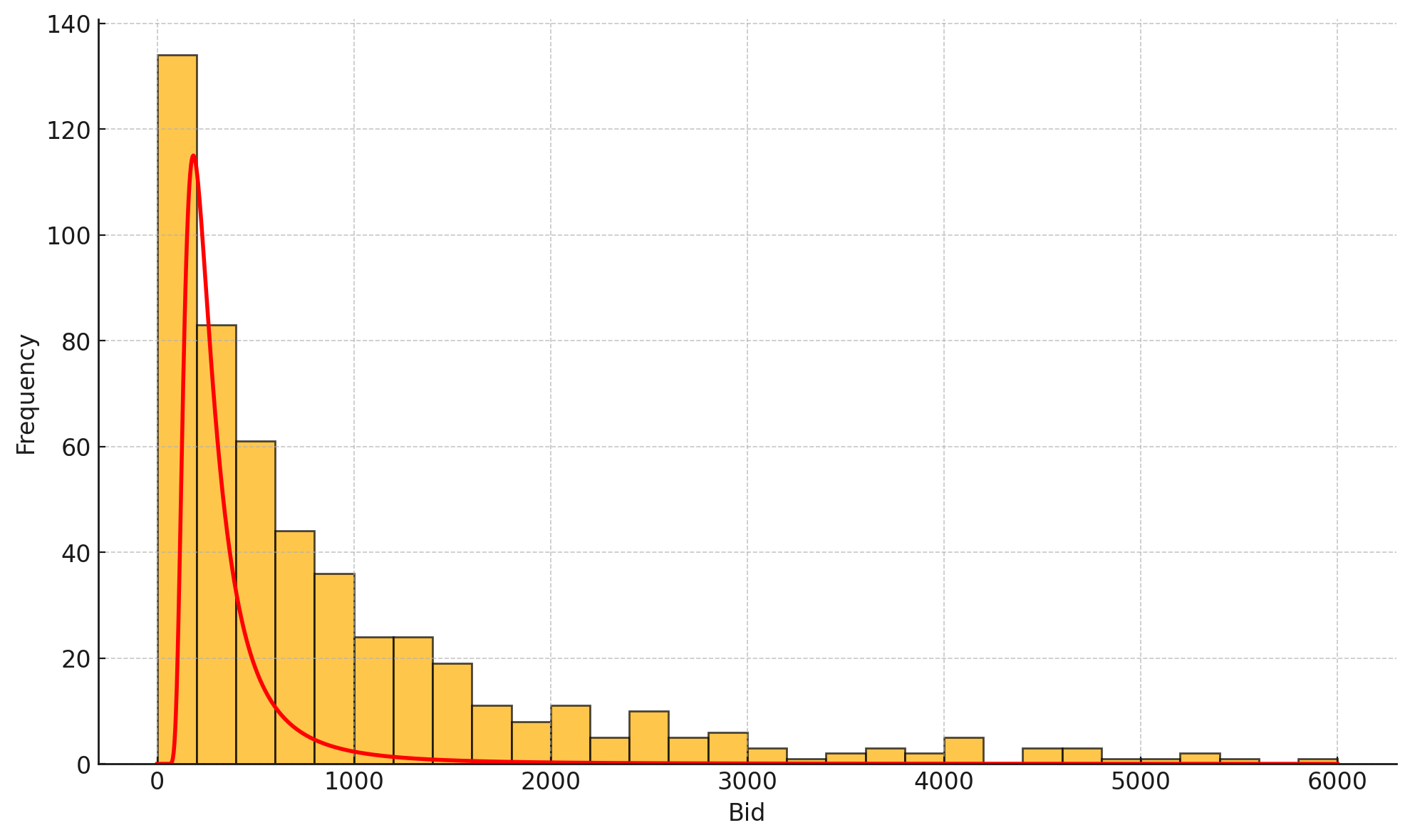}
            \caption{Frequency histogram of simulated bidder data.}
            \label{fig:density_samples}
     \end{subfigure}
     \caption{Frequency histograms of real (\ref{fig:density_bids}) and simulated data (\ref{fig:density_samples}) representing the valuations of $509$ bidders for 7-day auctions of Cartier watches on eBay, organized into bins of width 200. The density of the Frechet distribution $Fr(m = 0, s = 289, \alpha = 2.24)$, the best fit for the real data, is overlaid on top of both. The simulated data is drawn randomly from the same distribution.}
\end{figure}

\paragraph{Estimation and Extreme Value Analysis.} At this stage, we will use the data described above to identify an appropriate probability distribution that closely fits and accurately represents the probability density function (PDF) of consumer valuations.
We approximate the probability density function of consumer valuations by fitting a Fr\'{e}chet distribution to the data. The cumulative density function $F$ of $X \sim \normalfont{\text{Fr\'{e}chet}}(m, \sigma, \alpha)$ is given by
\[
F(x)= \exp\prn{-\prn{\frac{x-m}{s}}^{-\alpha}},
\] 
where $m$ is a location parameter representing the left-endpoint of the distribution,  $\sigma > 0 $ is a scale parameter, and $\alpha$ is a shape parameter. 

First, since our data is non-negative and includes bids near $0$, our estimate for $\hat{m}$ is $0$. Next, to estimate the shape parameter $\alpha$, we use Hill's estimator \cite{hill-estimator-frechet} which estimates $\alpha$ from the top $k$ order statistics for a given $k$. Our choice is $k = 97$, which is when Hill's estimator stabilizes to a value of $\hat{\alpha} = 2.24$. Finally, to estimate the scale $s$, we use the method of moments. Let $\hat{x}$ denote the sample mean and $\hat{\sigma}^2$ the sample variance. We can estimate $s$ from the sample mean and from the sample variance as
\[
\hat{s} = \frac{\hat{x}}{\Gamma\left(1 - \frac{1}{\hat{\alpha}}\right)} \qquad \text{ and } \qquad \hat{s} = \sqrt{\frac{\hat{\sigma}^2}{\Gamma\left(1 - \frac{2}{\hat{\alpha}}\right) - \prn{\Gamma\left(1 - \frac{1}{\hat{\alpha}}\right)}^2}},
\]
respectively. Since the two estimates vary wildly, to minimize the error, we estimate $s$ using both moments simultaneously. For this, we define the loss function
\[
L(\hat{s}) = \left(\hat{s} \cdot \Gamma\left(1 - \frac{1}{\hat{\alpha}}\right) - \hat{x} \right)^2 + \left( \hat{s}^2 \cdot \left(\Gamma\left(1 - \frac{2}{\hat{\alpha}}\right) - \prn{\Gamma\left(1 - \frac{1}{\hat{\alpha}}\right)}^2 \right) - \hat{\sigma}^2 \right)^2,
\]
which we minimize to obtain an estimate of $\hat{s} = 289$.

Next, we will draw simulated data from the empirical distribution to illustrate the quality of the fit. We use $\hat{F} = Fr(m = 0, s = 289, \alpha = 2.24)$ to denote the Fr\'{e}chet distribution with the corresponding parameters. We draw $n = 509$ samples from $\hat{F}$, representing a maximum bid for each simulated bidder, and present them in a frequency histogram in Figure~\ref{fig:density_samples}. The frequency histogram of the simulated bids is a close match to the one for the real bids, indicating that the fit of the distribution is good.

\paragraph{Welfare and Revenue Guarantees.}

Now that we have identified the distribution that most accurately approximates the consumer valuations when the market is large, we aim to use this information to determine what percentage of the optimal mechanism’s welfare can be guaranteed by employing a simple fixed-price policy, and illustrate that our single threshold algorithm achieves this percentage on the real data.

To do this, we first need to compute the optimal $x = x_{\alpha}$ for $\alpha = 2.24$ in \eqref{eq:opt-U-for-1}. We denote this optimal $x$ by $U$ and note that it corresponds to the value of $U$ from Lemma~\ref{lem:frechet-technical-2}. This optimal value is $U \approx 0.849$, which implies that the guarantee implied by Theorem~\ref{thm:apx:multi} is $\apx_1(\hat{F}) \approx 73\%$ of the expected maximum valuation.
Furthermore, by Lemma~\ref{lem:frechet-technical-2}, our choice of threshold $T_n$ is
\[
T_n = U \cdot a_n = U \hat{F}^{-1}(1 - 1/n) = U \cdot \hat{s} \cdot \prn{\log\prn{\frac{n}{n-1}}}^{-1/\hat{a}},
\]
and thus, for $n = 509$ we get
\[
T_{509} = 0.849 \cdot 289 \cdot \prn{\log\prn{509/508}}^{-1/2.24} \approx 3962.5.
\]

Next, notice that setting this threshold yields a competitive ratio of at least $3962.5/5400 \approx 73.3\%$, which is higher than $\apx_1(\hat{F})$ already for a relatively small value of $n$, and significantly higher than the worst-case bound of $0.712$ across all instances.

\paragraph{Acknowledgements.} This work was partially funded by ANID (Chile) through grants Anillo ACT210005, Fondecyt 11240705, and Fondecyt 1241846, and by the Center for Mathematical Modeling (ANID FB210005). 
Some preliminary results of this article appeared in SAGT 2021 \cite{correa2021optimal}.

\bibliographystyle{acm}
{\small\bibliography{biblio_full}}

\begin{thebibliography}{10}

\bibitem{abdallah2025regime}
{\sc Abdallah, T., and Reed, J.}
\newblock Regime dependent approximations for the single-item dynamic pricing problem.
\newblock {\em Submitted for publication\/} (2025).

\bibitem{alaei2013simple}
{\sc Alaei, S., Fu, H., Haghpanah, N., and Hartline, J.}
\newblock The simple economics of approximately optimal auctions.
\newblock In {\em FOCS\/} (2013), pp.~628--637.

\bibitem{alaei2019optimal}
{\sc Alaei, S., Hartline, J., Niazadeh, R., Pountourakis, E., and Yuan, Y.}
\newblock Optimal auctions vs. anonymous pricing.
\newblock {\em Games and Economic Behavior 118\/} (2019), 494--510.

\bibitem{arnostma2020}
{\sc Arnosti, N., and Ma, W.}
\newblock Tight guarantees for static threshold policies in the prophet secretary problem.
\newblock {\em Operations Research 71}, 5 (2023), 1777--1788.

\bibitem{arsenis2021constrained}
{\sc Arsenis, M., Drosis, O., and Kleinberg, R.}
\newblock Constrained-order prophet inequalities.
\newblock In {\em Proceedings of the 2021 ACM-SIAM Symposium on Discrete Algorithms (SODA)\/} (2021), SIAM, pp.~2034--2046.

\bibitem{BGPS20}
{\sc Beyhaghi, H., Golrezaei, N., Leme, R.~P., P{\'a}l, M., and Sivan, B.}
\newblock Improved revenue bounds for posted-price and second-price mechanisms.
\newblock {\em Operations Research 69}, 6 (2021), 1805--1822.

\bibitem{beyhaghi2019optimal}
{\sc Beyhaghi, H., and Weinberg, S.~M.}
\newblock Optimal (and benchmark-optimal) competition complexity for additive buyers over independent items.
\newblock In {\em STOC\/} (2019), pp.~686--696.

\bibitem{bingham-rv}
{\sc Bingham, N.~H., Goldie, C.~M., and Teugels, J.~L.}
\newblock {\em Regular Variation}.
\newblock Encyclopedia of Mathematics and its Applications. Cambridge University Press, 1987.

\bibitem{correa-competition-complexity}
{\sc Brustle, J., Correa, J., Duetting, P., and Verdugo, V.}
\newblock The competition complexity of dynamic pricing.
\newblock {\em Mathematics of Operations Research 49}, 3 (2024), 1986--2008.

\bibitem{brustle2024competition}
{\sc Brustle, J., Correa, J., D{\"u}tting, P., Ezra, T., Feldman, M., and Verdugo, V.}
\newblock The competition complexity of prophet inequalities.
\newblock {\em Mathematics of Operations Research\/} (2025), forthcoming.

\bibitem{brustle2024splitting}
{\sc Brustle, J., Perez-Salazar, S., and Verdugo, V.}
\newblock Splitting guarantees for prophet inequalities via nonlinear systems.
\newblock {\em Mathematics of Operations Research\/} (2025), forthcoming.

\bibitem{bulow-klemperer}
{\sc Bulow, J., and Klemperer, P.}
\newblock Auctions versus negotiations.
\newblock {\em American Economic Review 86}, 1 (1996), 180--94.

\bibitem{cai2015extreme}
{\sc Cai, Y., and Daskalakis, C.}
\newblock Extreme value theorems for optimal multidimensional pricing.
\newblock {\em Games and Economic Behavior 92\/} (2015), 266--305.

\bibitem{CEGMM10}
{\sc Chakraborty, T., Even-Dar, E., Guha, S., Mansour, Y., and Muthukrishnan, S.}
\newblock Approximation schemes for sequential posted pricing in multi-unit auctions.
\newblock In {\em WINE\/} (2010), pp.~158--169.

\bibitem{Chawla2010}
{\sc Chawla, S., Hartline, J.~D., Malec, D.~L., and Sivan, B.}
\newblock Multi-parameter mechanism design and sequential posted pricing.
\newblock In {\em STOC\/} (2010), pp.~311--320.

\bibitem{choi1994medians}
{\sc Choi, K.~P.}
\newblock On the medians of gamma distributions and an equation of ramanujan.
\newblock {\em Proceedings of the American Mathematical Society 121}, 1 (1994), 245--251.

\bibitem{cohen2016invisible}
{\sc Cohen-Addad, V., Eden, A., Feldman, M., and Fiat, A.}
\newblock The invisible hand of dynamic market pricing.
\newblock In {\em EC\/} (2016), pp.~383--400.

\bibitem{CFHOV17}
{\sc Correa, J., Foncea, P., Hoeksma, R., Oosterwijk, T., and Vredeveld, T.}
\newblock Posted price mechanisms and optimal threshold strategies for random arrivals.
\newblock {\em Mathematics of Operations Research 46}, 4 (2021), 1452--1478.

\bibitem{CFPV19}
{\sc Correa, J., Foncea, P., Pizarro, D., and Verdugo, V.}
\newblock From pricing to prophets, and back!
\newblock {\em Operations Research Letters 47}, 1 (2019), 25--29.

\bibitem{correa2021optimal}
{\sc Correa, J., Pizarro, D., and Verdugo, V.}
\newblock Optimal revenue guarantees for pricing in large markets.
\newblock In {\em SAGT\/} (2021), Springer, pp.~221--235.

\bibitem{DFK16}
{\sc D{\"u}tting, P., Fischer, F., and Klimm, M.}
\newblock Revenue gaps for static and dynamic posted pricing of homogeneous goods.
\newblock {\em arXiv preprint arXiv:1607.07105\/} (2016).

\bibitem{eden2016competition}
{\sc Eden, A., Feldman, M., Friedler, O., Talgam{-}Cohen, I., and Weinberg, S.~M.}
\newblock The competition complexity of auctions: {A} bulow-klemperer result for multi-dimensional bidders.
\newblock In {\em EC\/} (2017), p.~343.

\bibitem{Ehsani2018}
{\sc Ehsani, S., Hajiaghayi, M., Kesselheim, T., and Singla, S.}
\newblock Prophet secretary for combinatorial auctions and matroids.
\newblock {\em SIAM Journal on Computing 53}, 6 (2024), 1641--1662.

\bibitem{ezra2024competition2}
{\sc Ezra, T., and Garbuz, T.}
\newblock The competition complexity of prophet inequalities with correlations.
\newblock {\em To appear in EC\/} (2024).

\bibitem{ezra2024competition1}
{\sc Ezra, T., and Garbuz, T.}
\newblock The competition complexity of prophet secretary.
\newblock {\em arXiv preprint arXiv:2411.10892\/} (2024).

\bibitem{feldman201899}
{\sc Feldman, M., Friedler, O., and Rubinstein, A.}
\newblock 99\% revenue via enhanced competition.
\newblock In {\em EC\/} (2018), pp.~443--460.

\bibitem{feng2019optimal}
{\sc Feng, Y., Hartline, J.~D., and Li, Y.}
\newblock Optimal auctions vs. anonymous pricing: Beyond linear utility.
\newblock In {\em EC\/} (2019), pp.~885--886.

\bibitem{feng2024beyond}
{\sc Feng, Y., and Jin, Y.}
\newblock Beyond regularity: Simple versus optimal mechanisms, revisited.
\newblock {\em arXiv preprint arXiv:2411.03583\/} (2024).

\bibitem{fisher28}
{\sc Fisher, R.~A., and Tippett, L. H.~C.}
\newblock Limiting forms of the frequency distribution of the largest or smallest member of a sample.
\newblock In {\em Mathematical Proceedings of the Cambridge Philosophical Society\/} (1928), vol.~24, Cambridge University Press, pp.~180--190.

\bibitem{gautschi}
{\sc Gautschi, W.}
\newblock Some elementary inequalities relating to the gamma and incomplete gamma function.
\newblock {\em Journal of Mathematical Physics 38}, 1 (1959), 77--81.

\bibitem{giannakopoulos2021optimal}
{\sc Giannakopoulos, Y., Po{\c{c}}as, D., and Zhu, K.}
\newblock Optimal pricing for mhr and $\lambda$-regular distributions.
\newblock {\em ACM Transactions on Economics and Computation 9}, 1 (2021), 1--28.

\bibitem{GM66}
{\sc Gilbert, J.~P., and Mosteller, F.}
\newblock Recognizing the maximum of a sequence.
\newblock {\em Journal of the American Statistical Association 61\/} (1966), 35--76.

\bibitem{gnedenko1943distribution}
{\sc Gnedenko, B.}
\newblock Sur la distribution limite du terme maximum d'une serie aleatoire.
\newblock {\em Annals of Mathematics\/} (1943), 423--453.

\bibitem{goldenshluger2022optimal}
{\sc Goldenshluger, A., and Zeevi, A.}
\newblock Optimal stopping of a random sequence with unknown distribution.
\newblock {\em Mathematics of Operations Research 47}, 1 (2022), 29--49.

\bibitem{dehaan-ferreira-evt}
{\sc Haan, L., and Ferreira, A.}
\newblock {\em Extreme Value Theory: An Introduction}.
\newblock Springer, 2006.

\bibitem{HKS07}
{\sc Hajiaghayi, M.~T., Kleinberg, R., and Sandholm, T.}
\newblock Automated online mechanism design and prophet inequalities.
\newblock In {\em AAAI\/} (2007), vol.~7, pp.~58--65.

\bibitem{hartline2015non}
{\sc Hartline, J.~D., and Lucier, B.}
\newblock Non-optimal mechanism design.
\newblock {\em American Economic Review 105}, 10 (2015), 3102--24.

\bibitem{HR09}
{\sc Hartline, J.~D., and Roughgarden, T.}
\newblock Simple versus optimal mechanisms.
\newblock In {\em EC\/} (2009), pp.~225--234.

\bibitem{hill-estimator-frechet}
{\sc Hill, B.~M.}
\newblock {A Simple General Approach to Inference About the Tail of a Distribution}.
\newblock {\em The Annals of Statistics 3}, 5 (1975), 1163 -- 1174.

\bibitem{HK82}
{\sc Hill, T.~P., and Kertz, R.~P.}
\newblock Comparisons of stop rule and supremum expectations of i.i.d. random variables.
\newblock {\em The Annals of Probability 10}, 2 (1982), 336--345.

\bibitem{jiang2022tightness}
{\sc Jiang, J., Ma, W., and Zhang, J.}
\newblock Tightness without counterexamples: A new approach and new results for prophet inequalities.
\newblock In {\em EC 2023\/} (2023), p.~909.

\bibitem{JJLZ21}
{\sc Jin, Y., Jiang, S., Lu, P., and Zhang, H.}
\newblock Tight revenue gaps among multi-unit mechanisms.
\newblock In {\em EC\/} (2021), pp.~654--673.

\bibitem{JLQTX19}
{\sc Jin, Y., Lu, P., Qi, Q., Tang, Z.~G., and Xiao, T.}
\newblock Tight approximation ratio of anonymous pricing.
\newblock In {\em STOC\/} (2019), pp.~674--685.

\bibitem{jin2020tight}
{\sc Jin, Y., Lu, P., Tang, Z.~G., and Xiao, T.}
\newblock Tight revenue gaps among simple mechanisms.
\newblock {\em SIAM Journal on Computing 49}, 5 (2020), 927--958.

\bibitem{KK91}
{\sc Kennedy, D.~P., and Kertz, R.~P.}
\newblock The asymptotic behavior of the reward sequence in the optimal stopping of iid random variables.
\newblock {\em The Annals of Probability\/} (1991), 329--341.

\bibitem{K86}
{\sc Kertz, R.~P.}
\newblock Stop rule and supremum expectations of i.i.d. rando mvariables: A complete comparison by conjugate duality.
\newblock {\em Journal of Multivariate Analysis 19\/} (1986), 88--112.

\bibitem{KW12}
{\sc Kleinberg, R., and Weinberg, S.~M.}
\newblock Matroid prophet inequalities and applications to multi-dimensional mechanism design.
\newblock {\em Games and Economic Behavior 113\/} (2019), 97--115.

\bibitem{LLR12}
{\sc Leadbetter, M.~R., Lindgren, G., and Rootz{\'e}n, H.}
\newblock {\em Extremes and related properties of random sequences and processes}.
\newblock Springer Science \& Business Media, 2012.

\bibitem{liu2018competition}
{\sc Liu, S., and Psomas, C.-A.}
\newblock On the competition complexity of dynamic mechanism design.
\newblock In {\em SODA\/} (2018), pp.~2008--2025.

\bibitem{livanos2024minimizationSODA}
{\sc Livanos, V., and Mehta, R.}
\newblock Minimization is harder in the prophet world.
\newblock In {\em SODA\/} (2024), pp.~424--461.

\bibitem{livanos2024minimization}
{\sc Livanos, V., and Mehta, R.}
\newblock Minimization iid prophet inequality via extreme value theory: A unified approach.
\newblock {\em To appear in EC\/} (2025).

\bibitem{molina2025prophet}
{\sc Molina, M., Gast, N., Loiseau, P., and Perchet, V.}
\newblock Prophet inequalities: Competing with the top $\ell$ items is easy.
\newblock In {\em SODA\/} (2025), pp.~1270--1307.

\bibitem{libroextremo}
{\sc Resnick, S.~I.}
\newblock {\em Extreme values, regular variation and point processes}.
\newblock Springer, 2013.

\bibitem{SM02}
{\sc Saint-Mont, U.}
\newblock A simple derivation of a complicated prophet region.
\newblock {\em Journal of Multivariate Analysis 80\/} (2002), 67--72.

\bibitem{SC83}
{\sc Samuel-Cahn, E.}
\newblock Comparisons of threshold stop rule and maximum for independent nonnegative random variables.
\newblock {\em The Annals of Probability 12}, 4 (1983), 1213--1216.

\bibitem{Y11}
{\sc Yan, Q.}
\newblock Mechanism design via correlation gap.
\newblock In {\em SODA\/} (2011), pp.~710--719.

\end{thebibliography}


\newpage

\appendix

\section{Welfare Guarantees: Proof of Theorem \ref{thm:apx:multi}} \label{sec:proof_main_thm}

In this section, we prove Theorem \ref{thm:apx:multi}, and we introduce some necessary technical results.
The following proposition gives an equivalent expression for the value $\apx_k(F)$, useful in our analysis.

\begin{proposition}
	\label{prop:numerator-multi}
	Let $F$ be a distribution. 
	Then, for every positive integer $k$ we have  
	\begin{equation}
		\label{eqn:comp:ratio2:k}
		\apx_k(F)= \liminf_{n\to \infty}\sup_{T \in \RR_+} \EE \left(X \vert X > T\right)\frac{\sum_{j=1}^{k} \PP(M_n^{j}>T)}{\sum_{j=1}^k \EE(M_n^{j})}.
	\end{equation}
\end{proposition}

\begin{proof}[Proof of Proposition \ref{prop:numerator-multi}]

Given $k$ and a value $T$, consider the random variable defined by $\calB^n_T\sim \text{Bin}(n, 1-F(T))$ and let $W=\min\{ k, \calB^n_{T}\}$.
Consider an i.i.d. sample $X_1,\ldots,X_n$ distributed according to $F$.
By conditioning we get that for each $i\in \{1,2,\ldots,n\}$,
{\small \begin{align*}
	\EE(X_i)&=\sum_{j=0}^{k-1} \EE ( X_i  |  W=j ) {{n}\choose{j}}(1-F(T))^j F(T)^{n-j} + \EE\left(X_i | W=k \right) \PP(\calB^n_T\ge k).
\end{align*}}
Consider the random variable $\calR^n_{k,T}$ equal to the summation of the first $W=\min\{ k, \calB^n_{T}\}$ values in the sample attaining a value larger than $T$.
Note that since the sample is i.i.d. we get that 
$\EE (\calR^n_{k,T}|  W=j ) =  j \cdot \EE \left(X \vert X > T \right)$
and since $\calB^n_T\sim \text{Bin}(n, 1-F(T_n))$ we have that 
\[\sum_{j=0}^{k-1} j{{n}\choose{j}}(1-F(T))^j F(T)^{n-j} + k \cdot \PP(\calB^n_T\geq k)= \EE(\min\{k, \calB^n_{T} \}).\] 
Putting all together we obtain $\textstyle \EE(\calR^n_{k,T})= \EE \left(X_1 \vert X_1 > T \right)\EE(\min\{k, \calB^n_{T})).$
Observe that since $\calB^n_{T}$ is a positive random variable for every positive integer $n$ and every $T$, we have that 
\[\EE(\min\{k,\calB^n_{T}\})=\sum_{j=0}^{k-1} \PP(\min\{k,\calB^n_{T}\}>j)=\sum_{j=0}^{k-1} \PP(M_n^{j+1}>T)=\sum_{j=1}^{k} \PP(M_n^{j}>T),\]
where the second equality holds since for $j\le k-1$ we have $W>j$ if and only if $M_n^{j+1}>T$.
We conclude that $\EE(\calR^n_{k,T})=\EE \left(X_1 \vert X_1 > T \right)\sum_{j=1}^{k} \PP(M_n^{j}>T).$
Finally, the proposition holds by using the later equality in \eqref{eqn:comp:ratio:k:def}.
\end{proof}

\subsection{Extreme Value Theory: Technical Preliminaries}\label{app:EVT-prelim}

In this subsection, we briefly summarize the main tools we use from extreme value theory in our analysis.
We start by formally stating the extreme value theorem.
\begin{theorem}[see, e.g.,~\cite{libroextremo}]
	\label{thm:fisher}
	Let $F$ be a distribution for which there exists a positive real sequence $\{a_n\}_{n\in \NN}$ and other sequence $\{b_n\}_{n\in \NN}$ such that $(M_n-b_n)/a_n$ converges in distribution to a random variable with distribution $H$, namely, $\PP\left(M_n-b_n\le a_n t\right)=F^n(a_nt+b_n)\to H(t)$
	for every $t\in \RR$ when $n\to \infty$.
	Then we have that one of the following possibilities holds:  $H$ is the Gumbel, $H$ is in the Fr\'echet family, or $H$ is in the Reverse Weibull family.
\end{theorem}
In our analysis, we also need a tool from extreme value theory related to the order statistics of a sample according to $F$.
We denote the order statistics of a sample of size $n$ by $M_n=M_n^1\ge M_n^2\ge \cdots\ge M_n^n$. 
The following result can be seen as a generalization of the previous fundamental theorem.
\begin{theorem}[see, e.g.,~\cite{LLR12}]
	\label{thm:leadbetter}
	Let $F$ be a distribution satisfying the extreme value condition with the scaling and shifting sequences $\{a_n\}_{n\in \NN}$ and $\{b_n\}_{n\in \NN}$ such that $\PP\left(M_n-b_n\le a_n t\right)\to H(t)$
	for every $t\in \RR$ when $n\to \infty$.
	Then, for each $j\in \{1,2,\ldots,n\}$ and every $t\in \RR$ we have 
	\[\lim_{n\to \infty}\PP\left(M_n^j-b_n\le a_n t\right)= H(t) \sum_{s=0}^{j-1} \frac{(-\log H(t))^s}{s!}.\]	
\end{theorem}

Recall that a distribution $V$ is a {\it Von Mises function} if there exist $z_0\in \RR$, a constant $\theta>0$ and a function $\aux:(\omega_0(V),\infty)\to \RR_+$ absolutely continuous with $\lim_{u \to \infty} \aux'(u)=0$, such that for every $t\in (z_0,\infty)$ we have
\begin{equation}\label{V-M:repr}
	1-V(t)=\theta \exp \Big( - \int\limits_{z_0}^t \frac{1}{\aux(s)} \mathrm ds  \Big).
\end{equation}
We call such $\aux$ an {\it auxiliary function} of $V$.
Next, we summarize some technical results related to Von Mises functions that we use in our analysis.
\begin{lemma}[see, e.g.,~\cite{libroextremo}]
	\label{lem:von-mises-prop}
	Let $V$ be a Von Mises function with auxiliary function $\aux$ and such that $\omega_1(V)=\infty$.
	Then, $V$ has extreme type Gumbel, and the shifting and scaling sequences may be chosen respectively as
	$b_n = V^{-1}(1-1/n)$ and $ a_n = \aux(b_n)$ for every $n$.
	Furthermore, we have $\lim_{t \to \infty}  \aux(t)/t = 0$ and $\lim_{t\to \infty}(t+x\mu(t))=\infty$ for every $x\in \RR$.
\end{lemma}
For example, the exponential distribution of parameter $\lambda$ is a Von Mises function, with auxiliary constant function $1/\lambda$, $\theta=1$ and $z_0=0$.
Furthermore, for every positive integer $n$ we have $b_n=F^{-1}(1-1/n)=(\log n)/\lambda$ and $a_n=\aux(b_n)=1/\lambda$.
A relevant property states that every distribution with extreme type Gumbel can be represented by a distribution that is a Von Mises function in the following precise sense.
\begin{lemma}[see, e.g.,~\cite{libroextremo}]
	\label{prop:Gumbel-prop}
	Let $F$ be a distribution satisfying the extreme value condition with $\omega_1(F)=\infty$.
	Then, $F$ has extreme type Gumbel if and only if there exists a Von Mises function $V$ and a positive function $\eta:(\omega_0(F),\infty)\to \RR_+$ with $\lim_{t \to \infty} \eta(t)= \eta^{\star}>0$ such that $1-F(t)=\eta(t)(1-V(t))$ for every $t\in (\omega_0(F),\infty)$. 
\end{lemma}
Thus, when $F$ has extreme type Gumbel, there exists a pair $(V,\eta)$ satisfying the condition guaranteed in Lemma \ref{prop:Gumbel-prop}, and we say that $(V,\eta)$ is a {\it Von Mises representation} of $F$.
Furthermore, when $\lim_{t \to \infty} \eta'(t)=0$, we say that $(V,\eta)$ {\it smoothly represents} $F$.

\subsection{Proof of Theorem \ref{thm:apx:multi}\ref{multiple-frechet}: The Fr\'echet Family}\label{sec:proof:thm}

In what follows we restrict to the case in which the distribution $F$ has extreme type in the Fr\'echet family.
We remark that if $\alpha\in (0,1]$ the expected value of a random variable with distribution Fr\'echet $\Phi_{\alpha}$ is not finite.
Therefore, we further restrict to the Fr\'echet family where $\alpha\in (1,\infty)$. 
To prove Theorem \ref{thm:apx:multi}\ref{multiple-frechet} we require a technical lemma, where we exploit the structure given by the existence of an extreme value and we show how to characterize the approximation factor of a distribution in the Fr\'echet family for large values of $n$.
Before stating this lemma, let us introduce a few results about the Fr\'echet family that will be required.

We say that a positive measurable function $\ell:(0,\infty)\to \mathbb{R}$ is {\it slowly varying} if for every $u>0$ we have 
$\lim_{t\to \infty}\ell(ut)/\ell(t)=1.$
For example, the function $\ell(t)=\log(t)$ is slowly varying, since $\ell(ut)/\ell(t)=\log(u)/\log(t)+1\to 1$ when $t\to \infty$.  
On the other hand, the function $\ell(t)=t^{\gamma}$ is not slowly varying, since for every $u>0$ we have $\ell(ut)/\ell(t)=u^{\gamma}$.
The following lemma shows the existence of a strong connection between the distributions with extreme type in Fr\'echet family and slowly varying functions.
Recall that for $\alpha>0$, the Pareto distribution of parameter $\alpha$ is given by $P_{\alpha}(t)=1-t^{-\alpha}$ for $t\ge 1$ and zero otherwise.
\begin{lemma}[see, e.g., \cite{libroextremo}]
\label{lem:fre-key}
Let $F$ be a distribution with extreme type in the Fr\'echet family. For every positive integer $n$, the sequences $a_n=F^{-1}(1-1/n)$ and $b_n=0$ are valid scaling and shifting sequences for $F$. \label{frechet-tech-1}
Furthermore, there exists a slowly varying function $\ell_F$ such that $1-F(t) = (1-P_{\alpha}(t))\cdot \ell_F(t) = t^{-\alpha}\ell_F(t)$, for every $t\in \RR_+$.
\end{lemma}

Observe that this lemma says that if $F$ has extreme type Fr\'echet of parameter $\alpha$, then it essentially corresponds to a perturbation of a Pareto distribution with parameter $\alpha$ by some slowly varying function.
We are now ready to state the main technical lemma.



 \begin{lemma}
 \label{lem:frechet-technical-2}
 Let $F$ be a distribution with extreme type Fr\'echet of parameter $\alpha$ and let $\{a_n\}_{n\in \NN}$ be an appropriate scaling sequence. 
 Consider a positive sequence $\{T_n\}_{n\in \NN}$ with $T_n\to \infty$ and for which there exists $U\in \RR_+$ such that $T_n/a_n\to U$, as $n$ goes to infinity.
 Then, we have 
{ \small \[\lim_{n\to \infty}\EE \left(X \vert X > T_n\right)\frac{\sum_{j=1}^{k} \PP(M_n^{j}>T_n)}{\sum_{j=1}^k \EE(M_n^{j})}= \frac{\Gamma(k)}{\Gamma(k+1-1/\alpha)} U \exp(-U^{-\alpha}) \sum_{j=1}^k \sum_{s=j}^{\infty} \frac{U^{-s \alpha}}{s!}.\]}
 \end{lemma}
We use this lemma to prove Theorem~\ref{thm:apx:multi}\ref{multiple-frechet}, and then we prove the lemma. 

\begin{proof}[Proof of Theorem \ref{thm:apx:multi}\ref{multiple-frechet}]
 Let $F$ be a distribution with extreme type Fr\'echet of parameter $\alpha$. 
 We first prove that for each positive integer $k$ it holds that $\apx_k(F) \ge \varphi_k(\alpha)$. 
To this end, for each positive integer $n$ and positive real number $U$, let $T_n$ be the threshold given by $T_n=a_n \cdot U$, where $\{ a_n\}_{n \in \mathbb{N}}$ is the scaling sequence for the distribution $F$ given by Lemma~\ref{lem:fre-key}.
Then, 
\begin{eqnarray}\label{ineq:apx}
\apx_k(F) \ge \liminf_{n\to \infty} \EE \left(X \vert X > T_n \right)\frac{\sum_{j=1}^{k} \PP(M_n^{j}>T_n)}{\sum_{j=1}^k \EE(M_n^{j})}. 
\end{eqnarray}

Note that $\lim \inf_{n \to \infty} T_n= \infty$ (and thus $T_n \to \infty$), since $U \in \RR_+$ and $a_n \to \infty$. Furthermore, $\lim_{n \to \infty} T_n/a_n= U$ and then 
applying Lemma~\ref{lem:frechet-technical-2} together with inequality  \eqref{ineq:apx} we obtain that
\[\apx_k(F) \ge  \frac{\Gamma(k)}{\Gamma(k+1-1/\alpha)} U \exp(- U^{-\alpha}) \sum_{j=1}^k \sum_{s=j}^{\infty} \frac{U^{-s \alpha}}{s!}.\]

Given that the inequality above holds for every positive real number $U$, we have 
\[\apx_k(F) \ge  \frac{\Gamma(k)}{\Gamma(k+1-1/\alpha)} \max_{U \in \RR_+} U \exp(- U^{-\alpha}) \sum_{j=1}^k \sum_{s=j}^{\infty} \frac{U^{-s \alpha}}{s!}= \varphi_k (\alpha).\] 

In the rest of the proof we show that, for each positive real number $k$ and each $\alpha \in (1, \infty)$, $\varphi_k(\alpha)$ is lower bounded by $ 1-1/\sqrt{2k \pi}.$ 
To this end, we just need to evaluate the objective function of our optimization problem in a well chosen value.
An inequality for the Gamma function due to Gautschi~\cite{gautschi} states that for every $s \in (0,1)$ and every $x \ge 1$ we have $\Gamma(x+1)>x^{1-s}\cdot \Gamma(x+s)$.
Then, setting $x=k$ and $s=1-1/\alpha$ yields $\Gamma(k+1)>k^{1/\alpha} \Gamma(k+1-1/\alpha).$
Since $\Gamma(k)=\Gamma(k+1)/k$, we therefore obtain
$ k^{1-1/\alpha}>\Gamma(k+1-1/\alpha)/\Gamma(k).$
On the other hand, note that for each $U \in (0, \infty)$ we have 
\begin{eqnarray*}
	U \exp(-U^{-\alpha})\sum_{j=1}^k\sum_{s=j}^\infty \frac{U^{-s \alpha}}{s!} 
	&=& U \exp(-U^{-\alpha})\left( \sum_{s=1}^k \frac{s\cdot  U^{-s \alpha}}{s!}+ k \sum_{s=k+1}^\infty \frac{U^{-s\alpha}}{s!}\right) \\
	&=&U \exp(-U^{-\alpha})\left( U^{-\alpha} \sum_{s=0}^{k-1} \frac{U^{-s \alpha}}{s!}+ k \sum_{s=k+1}^\infty \frac{U^{-s\alpha}}{s!}\right).
\end{eqnarray*}
In particular, by taking $U_{k,\alpha}=k^{-1/\alpha}$ we get that
{\small \begin{eqnarray*}
	\varphi_k(\alpha)\cdot \frac{\Gamma(k+1-1/\alpha)}{\Gamma(k)} &\ge& U_{k,\alpha}\cdot  k \exp(-U_{k,\alpha}^{-\alpha}) \left(\sum_{s=0}^{k-1} \frac{U_{k,\alpha}^{-s \alpha}}{s!}+ \sum_{s=k+1}^\infty \frac{U_{k,\alpha}^{-s\alpha}}{s!}\right) \\
	&=&  U_{k,\alpha}\cdot  k \exp(-U_{k,\alpha}^{-\alpha}) \left(\exp(U_{k,\alpha}^{-\alpha})- \frac{U_{k,\alpha}^{-\alpha k}}{k!}\right) \\
	&= & k^{1-1/\alpha} \left( 1- \frac{e^{-k}k^k}{k!}\right)\ge \frac{\Gamma(k+1-1/\alpha)}{\Gamma(k)}\left(1- \frac{1}{\sqrt{2 \pi k}} \right),
\end{eqnarray*}}
where the first inequality follows since the value of $\varphi_k(\alpha)$ involves the maximum over $(0,\infty)$, the first equality from the Taylor series for the exponential function and the last inequality is obtained by applying Stirling's approximation inequality.
This concludes the proof of the theorem.
\end{proof}

Below, we show a result that we need to prove Lemma~\ref{lem:frechet-technical-2}.
\begin{lemma}
	\label{lem:fre-integral}
	Let $F$ be a distribution with extreme type distribution Fr\'echet of parameter $\alpha\in (1,\infty)$ and let $\{T_n\}_{n\in \NN}$ be a sequence of real values such that $T_n\to \infty$.
	Then, we have that
	\begin{equation*}
		\lim_{n\to \infty}\frac{\EE(X|X> T_n)}{T_n}= \frac{\alpha}{\alpha-1}.
	\end{equation*}
\end{lemma}
\begin{proof}[Proof of Lemma \ref{lem:fre-integral}
]
	To prove the lemma, recall the classic dominated convergence theorem: Suppose that $\{ f_n\}_{n\in \NN}$ is a sequence of measurable functions such that $f_n \to f$ pointwise almost everywhere as $n \to \infty$. 
	If there exists an integrable function $g$ satisfying $|f_n|\leq g$ for every positive integer $n$, then f is also integrable and $\int f \mathrm du= \lim_{n \to \infty} \int f_n \mathrm du.$
	We first observe that by expanding the conditional expectation we get that 
	\[\frac{\EE(X|X> T_n)}{T_n}=1+\frac{1}{T_n(1-F(T_n))}\int_{T_n}^{\infty}(1-F(s))\mathrm ds\]
	for every positive integer $n$.
	By Lemma~\ref{lem:fre-key}, there exists a slowly varying function $\ell_F$ such that $1-F(t)=t^{-\alpha}\ell_F(t)$ for every $t\in \RR_+$.
	Therefore, 
	\begin{align*}
		\frac{1}{T_n(1-F(T_n))}\int_{T_n}^{\infty}(1-F(s))\mathrm ds&=\frac{1}{T_n\cdot T_n^{-\alpha}\cdot \ell_F(T_n)}\int_{T_n}^{\infty}s^{-\alpha}\ell_F(s)\mathrm ds \\
		&=\frac{1}{T_n^{1-\alpha}}\int_{T_n}^{\infty}s^{-\alpha}\frac{\ell_F(s)}{\ell_F(T_n)}\mathrm ds.
	\end{align*}
	We perform a change of variables in the integral by setting $s=u\cdot T_n$.
	Then, we obtain that the integral above is equal to
{\small	\begin{align*}
		\frac{1}{T_n^{1-\alpha}}\int_{T_n}^{\infty}s^{-\alpha}\frac{\ell_F(s)}{\ell_F(T_n)}\mathrm ds&=\frac{1}{T_n^{1-\alpha}}\int_{1}^{\infty}u^{-\alpha}T_n^{-\alpha}\frac{\ell_F(uT_n)}{\ell_F(T_n)} \cdot T_n\mathrm du=\int_{1}^{\infty}u^{-\alpha}\frac{\ell_F(uT_n)}{\ell_F(T_n)}\mathrm du.
	\end{align*}}
	Since the function $\ell_F$ is slowly varying and $T_n\to \infty$ we have that for every $u\in [1,\infty)$ it holds  $u^{-\alpha}\ell_F(uT_n)/\ell_F(T_n)\to u^{-\alpha},$
	when $n\to \infty$.
	Furthermore, for every $\alpha\in (1,\infty)$ the function $u^{-\alpha}$ is integrable in $(1,\infty)$ and thus by the dominated convergence theorem we have that 
	\begin{equation*}
		\lim_{n\to \infty}\int_{1}^{\infty}u^{-\alpha}\frac{\ell_F(uT_n)}{\ell_F(T_n)}\mathrm du = \int_{1}^{\infty}u^{-\alpha}\mathrm du=\frac{1}{\alpha-1},
	\end{equation*}
	which finishes the proof of the lemma.	
\end{proof}

\begin{proof}[Proof of Lemma \ref{lem:frechet-technical-2}]
	First, recall the following fact: If we have a sequence of probability distributions $\{G_n\}_{n\in \NN}$ such that $G_n$ converges point-wise to a probability distribution $G$, then $G_n(z_n)\to G(z)$ when $z_n\to z$.
	Consider a positive sequence $\{T_n\}_{n\in \NN}$ such that the following holds: $T_n\to \infty$ and there exists $U\in \RR_+$ such that $T_n/a_n\to U$.
	Since $F$ has extreme type distribution Fr\'echet of parameter $\alpha$, by Theorem~\ref{thm:leadbetter} we have that  
	 \begin{eqnarray*}
		\lim_{n \to \infty} \PP(M_n^{j}/a_n \le T_n/a_n) 
		&=&  \exp(-U^{-\alpha})\sum_{s=0}^{j-1} \frac{(-\log(\exp(-U^{-\alpha})))^s}{s!} \\
		&=& \exp(-U^{-\alpha})\sum_{s=0}^{j-1} \frac{U^{-\alpha s}}{s!},
	\end{eqnarray*}
	and therefore
	\[
	\lim_{n \to \infty} \sum_{j=1}^{k} \PP(M_n^{j}>T_n)=\exp(-U^{-\alpha})\sum_{j=1}^k\sum_{s=j}^{\infty} \frac{U^{-\alpha s}}{s!},
	\]
	where the equality follows since $\sum_{s=0}^{\infty} U^{-\alpha s}/s!=\exp(U^{-\alpha})$. 
	In the remainder of the proof, we analyze the limit of $T_n/\sum_{j=1}^k\EE(M_n^j)$ when $n\to \infty$, which is equivalent to analyzing the sequence $(T_n/a_n)/\sum_{j=1}^k \EE(M_n^j/a_n)$.
	To this end, we first note that $\{M_n^j/a_n \}_{n\in \NN}$ converges not only in distribution but also in expectation to a random variable distributed according to 
	\[G_j(t)=H(t) \sum_{s=0}^{j-1} \frac{(-\log H(t))^s}{s!},\] 
	with $H(t)=\exp(-t^{-\alpha})$. 
	More specifically, $\lim_{n \to \infty} \EE(M_n^j/a_n)= \EE(Z_j)$
	where $Z_j$ is distributed according to  $G_j.$ 
	This follows from a result on the convergence of moments in extreme value theory~\cite[p. 77, Proposition 2.1]{libroextremo}.
	Furthermore, 
	\begin{align*}
		\EE(Z_j)&=\int_{0}^\infty t \;\mathrm{d}G_j(t) \\
		&= \int_{0}^\infty \exp(-t^{-\alpha})\cdot \alpha \sum_{s=0}^{j-1} \frac{t^{-\alpha(s+1)}}{s!} \mathrm{d}t - \int_{0}^\infty \exp(-t^{-\alpha}) \cdot \alpha \sum_{s=1}^{j-1} \frac{t^{-\alpha s}}{(s-1)!} \mathrm{d}t \\
		&= \int_{0}^\infty \alpha \exp(-t^{-\alpha}) \frac{t^{-\alpha j}}{(j-1)!} \mathrm{d}t\\ 
		&= \frac{1}{\Gamma(j)}  \int_{0}^\infty \alpha \exp(-t^{-\alpha}) \cdot t^{-\alpha j} \mathrm{d}t \\
&= \frac{1}{\Gamma(j)}\int_{0}^\infty e^{-x}  x^{j-1-1/\alpha} \mathrm{d}x = \frac{\Gamma(j-1/\alpha)}{\Gamma(j)},
	\end{align*}
	where the fourth equality holds because for every natural number $\Gamma(x)=(x-1)!,$ the fifth by performing the change of variables $x=t^{-\alpha}$ and the last one follows from the definition of the Gamma function.
	\begin{claim} \label{claim:sum:gamma}
		For every positive integer $k$ and real value $\alpha>1$, it holds that
		\begin{equation}\label{eq:sum:gammas}
			\sum_{j=1}^k \frac{\Gamma(j-1/\alpha)}{\Gamma(j)}=\frac{\alpha}{\alpha-1} \cdot \frac{\Gamma(k+1-1/\alpha)}{\Gamma(k)}. 
		\end{equation}
	\end{claim}	
    \begin{proof}
    We prove the claim by induction on $k$. 
    Note that when $k=1$ the claim follows directly by observing that $\Gamma(2-1/\alpha)=(1-1/\alpha)\cdot \Gamma(1-1/\alpha)$.
    When $k>1$, by manipulating the left hand side of  \eqref{eq:sum:gammas} we obtain 
	\begin{align*}
		\sum_{j=1}^k \frac{\Gamma(j-1/\alpha)}{\Gamma(j)}&= \sum_{j=1}^{k-1} \frac{\Gamma(j-1/\alpha)}{\Gamma(j)}+ \frac{\Gamma(k-1/\alpha)}{\Gamma(k)}\\
		&=\frac{1}{1-1/\alpha} \cdot \frac{\Gamma(k-1/\alpha)}{\Gamma(k-1)}+  \frac{\Gamma(k-1/\alpha)}{\Gamma(k)} \\
		&= \frac{\Gamma(k-1/\alpha)}{\Gamma(k)} \left( \frac{k-1}{1-1/\alpha}+1\right)\\
		&= \frac{\Gamma(k-1/\alpha)}{\Gamma(k)} \cdot  \frac{k-1/\alpha}{1-1/\alpha} =  \frac{\alpha}{\alpha-1} \cdot \frac{\Gamma(k+1-1/\alpha)}{\Gamma(k)},
	\end{align*}
	where the second equality follows by applying the induction hypothesis and the third and the last equalities follow from the identity $\Gamma(x)\cdot x=\Gamma(x+1)$ for $x=k-1$ and $x=k-1/\alpha$, respectively.
    \end{proof}
	Using the above claim, along with the fact that $T-n/a_n \to U$ as $n$ goes to infinity, we have that 
	\[\lim_{n\to \infty}\frac{T_n/a_n}{\sum_{j=1}^k \EE(M_n^j/a_n)}= U\cdot \frac{\alpha-1}{\alpha} \cdot \frac{\Gamma(k)}{\Gamma(k+1-1/\alpha)}\]
	when $n\to \infty$.
	Thus, combining the above with Lemma~\ref{lem:fre-integral}, we get that 
	\begin{align*}
	&\EE \left(X \vert X > T_n\right)\frac{\sum_{j=1}^{k} \PP(M_n^{j}>T_n)}{\sum_{j=1}^k \EE(M_n^{j})}\\
	&=\frac{\EE \left(X \vert X > T_n\right)}{T_n}\cdot \frac{T_n/a_n}{\sum_{j=1}^k \EE(M_n^j/a_n)}\cdot \sum_{j=1}^{k} \PP(M_n^{j}>T_n)\\
	&\to \frac{\alpha}{\alpha-1}\cdot U\cdot \frac{\alpha-1}{\alpha} \cdot \frac{\Gamma(k)}{\Gamma(k+1-1/\alpha)}\exp(-U^{-\alpha}) \sum_{j=1}^k \sum_{s=j}^{\infty} \frac{U^{-s \alpha}}{s!}\\
	&=\frac{\Gamma(k)}{\Gamma(k+1-1/\alpha)} U \exp(-U^{-\alpha}) \sum_{j=1}^k \sum_{s=j}^{\infty} \frac{U^{-s \alpha}}{s!},
	\end{align*}
	when $n\to \infty$.
\end{proof}


\subsection{Proof of Theorem \ref{thm:apx:multi}\ref{multiple-Gumbel}: Gumbel and Reversed Weibull Family }
\label{sec:thm-gumbel-section}

In what follows we consider a distribution $F$ with extreme type Gumbel or in the reversed Weibull family.
In what follows, we restrict attention to the distributions $F$ with extreme type Gumbel where $\omega_1(F)=\infty$.
Key to our analysis are Lemmas~\ref{lem:von-mises-prop} and~\ref{prop:Gumbel-prop} about Von Mises representations for
distributions in the Gumbel family.
Before proving the theorem, we need a couple of lemmas regarding the structure of a distribution in the Gumbel family.

\begin{lemma}
	\label{lem:mises-to-gumbel}
	Let $F$ be a distribution with extreme type in the Gumbel family such that $\omega_1(F)=\infty$ and let $(V,\eta)$ be a Von Mises representation of $F$ such that $\lim_{t\to \infty}\eta(t)=\eta^{\star}$.
	Let $\{a_n\}_{n\in\NN}$ and $\{b_n\}_{n\in\NN}$ be scaling and shifting sequences, respectively, for $V$.
	For every positive integer $n$ consider $b_n^{\eta}=b_n+a_n\log \eta^{\star}$.
	Then, the following holds:
	\begin{enumerate}[label=\normalfont(\roman*)]
		\item $\{a_n\}_{n\in\NN}$ and $\{b_n^{\eta}\}_{n\in\NN}$ are scaling and shifting sequences, respectively, for $F$.\label{sequences-gumbel}
		\item For every $U\in \RR$ we have $\lim_{n\to \infty}(a_nU+b_n^{\eta})= \infty$. \label{limit-between}
		\item For every $U\in \RR$ and every positive integer $k$ we have that $\lim_{n\to \infty}(a_nU+b_n^{\eta})/\sum_{j=1}^k \EE(M_n^{j})=1/k$,
where $M_n^1,\ldots,M_n^n$ are the order statistics for $F$. \label{lem:gumbel-k-second}
	\end{enumerate}
\end{lemma}
\begin{proof}
	To prove the first point we start by checking that for every $t$ we have $\PP\left(M_n-b_n\le a_n\cdot t\right)=F^n(a_nt+b_n)\to \exp(- \eta^{\star} e^{-t})$. 
	For every positive integer $n$, let $\eta_n=\eta(a_nt+b_n)$ and $d_n=V(a_nt+b_n)$.
	Observe that, by Lemma~\ref{prop:Gumbel-prop}, $F^n(a_nt+b_n)=\left(1-\eta_n(1-d_n)n\cdot \frac{1}{n}\right)^n,$
	and therefore, it is enough to check that $\eta_n(1-d_n)n\to \eta^{\star} e^{-t}$.
	Recall that by Lemma~\ref{lem:von-mises-prop} we have that $a_n=\aux(b_n)$, where $\mu$ is the auxiliary function of $V$ and $b_n=\inf\{y\in \RR:V(y)\ge 1-1/n\}\to \infty$ (since $\omega_1(V)=\infty$).
	Therefore, we have that $a_nt+b_n=\mu(b_n)t+b_n\to \infty$, where the limit holds by Lemma~\ref{lem:von-mises-prop} for $x = t$ and $t = b_n$, since $b_n\to \infty$.  
	Furthermore, since $d_n \to 1$, we have  $\lim_{x\to 1}(1-x)^{-1}\log(1/x)=1$ and therefore $\lim_{n\to \infty}n(1-d_n)=\lim_{n\to \infty}n\log(1/d_n)=e^{-t}$, since $d_n^n\to \exp(-e^{-t})$. Thus, for every $t > 0$, we have $\lim_{n \to \infty} \eta_n(1 - d_n) n = \eta^{\star} e^{-t}$
	On the other hand, from the definition of the sequences, we have
	\begin{align*}
		\PP\left(M_n-b_n^{\eta}\le a_n\cdot t\right)&=\PP\left(M_n-b_n-a_n\log \eta^{\star}\le a_n\cdot t\right)\\
		&=\PP\left(M_n-b_n\le a_n(t+\log \eta^{\star})\right)\\ 
		&\to \exp\Big(- \eta^{\star} e^{-(t+\log \eta^{\star})}\Big)=\exp(-e^{-t}),
	\end{align*}
	which is the Gumbel distribution and that concludes \ref{sequences-gumbel}.
	To prove \ref{limit-between},
	note that for every $U$, we have that $a_nU+b_n^{\eta}=a_nU+b_n+a_n\log \eta^{\star}=\mu(b_n)(U+\log\eta^{\star})+b_n\to \infty$,  
where this limit holds since $b_n\to \infty$ and by Lemma~\ref{lem:von-mises-prop}. 

We now prove part \ref{lem:gumbel-k-second}.
We show that for every $j\in \{1,\ldots,k\}$ we have that $\EE(M_n^j/(a_nU+b_n^{\eta}))\to 1$ when $n\to \infty$.
By part \ref{sequences-gumbel} we have that $\{a_n\}_{n\in\NN}$ and $\{b_n^{\eta}\}_{n\in\NN}$ are scaling and shifting sequences, respectively, for $F$.
We first note that $\{(M_n^j-b^{\eta}_n)/a_n \}_{n\in \NN}$ converges not only in distribution but also in expectation to a random variable distributed according to 
\[G_j(t)=H(t) \sum_{s=0}^{j-1} \frac{(-\log H(t))^s}{s!},\] 
with $H(t)=\exp(- e^{-t})$. 
More specifically, $\lim_{n \to \infty} \EE((M_n^j-b^{\eta}_n)/a_n)= \EE(Z_j)$
where $Z_j$ is distributed according to  $G_j$.
This follows from a result on the convergence of moments in extreme value theory~\cite[Proposition 2.1, p. 77]{libroextremo}.
Furthermore, we have $\EE(Z_j)\le \EE(Z_1)<\infty$. 
We claim that $b^{\eta}_n/a_n\to \infty$.
Since $\omega_1(V)=\infty$ we have that $b_n=\inf\left\{t\in \RR:V(t)\ge 1-1/n\right\}\to \infty$ when $n\to \infty$.
Therefore, by Lemma~\ref{lem:von-mises-prop} it follows that
$b_n/a_n=b_n/\aux(b_n)\to \infty$ when $n\to \infty$, where $\mu$ is the auxiliary function of $V$, and therefore $b^{\eta}_n/a_n=b_n/\mu(b_n)+\log(\eta^{\star})\to \infty$, where the limit holds since $b_n/\mu(b_n)\to \infty$ by Lemma~\ref{lem:von-mises-prop} and the fact that $b_n\to \infty$. 
\begin{claim}
\label{claim:seq-limit}
Suppose we have three positive sequences of real values $\{z_n\}_{n\in \NN}$, $\{w_n\}_{n\in \NN}$ and $\{v_n\}_{n\in \NN}$ such that $(z_n-w_n)/v_n\to L<\infty$ and $w_n/v_n\to \infty$. 
Then, $z_n/v_n\to \infty$ and $z_n/w_n\to 1$.
\end{claim}
\begin{proof}
For every $M\ge L-1$, there exists a positive integer $n_0$ such that for every $n\ge n_0$ the following inequalities hold: $w_n>(M-L+1)v_n$ and $z_n-w_n>(L-1)v_n$.
Therefore, for every $n\ge n_0$ we have that $z_n>(M-L+1)v_n+(L-1)v_n=Mv_n$ and thus we conclude that $z_n/v_n\to \infty$.
For the other one, observe that there exists $n_1$ such that for every $n\ge n_1$ we have that $(L/2)v_n<z_n-w_n<(3L/2)v_n$.
Therefore, for every $n\ge n_1$ we have that $(L/2)\frac{v_n}{w_n}<z_n/w_n-1<(3L/2)\frac{v_n}{w_n}$ and since $v_n/w_n\to 0$ we conclude that $z_n/w_n\to 1$.
\end{proof}
Using Claim \ref{claim:seq-limit} with $z_n=\EE(M^j_n)$, $w_n=b^{\eta}_n$ and $v_n=a_n$ we conclude that $\EE(M^j_n)/a_n\to \infty$ and $\EE(M^j_n)/b^{\eta}_n\to 1$.
Then, we have that $U\cdot a_n/\EE(M^j_n)+b^{\eta}_n/\EE(M^j_n) \to 1$ for every $U\in \RR$.
\end{proof}

\begin{lemma}
	\label{lem:gumbel-first-term}
	Let $F$ be a distribution with extreme type in the Gumbel family and let $\{\Theta_n\}_{n\in \NN}$ be a sequence of real values such that $\Theta_n\to \infty$.
	Then, we have $\lim_{n\to \infty}\EE(X|X>\Theta_n)/\Theta_n=1$, where $X$ is distributed according to $F$.
\end{lemma}
\begin{proof}
	Let $(V,\eta)$ be a Von Mises representation of $F$ such that $\lim_{t\to \infty}\eta(t)=\eta^{\star}$, guaranteed by Lemma \ref{prop:Gumbel-prop}.
	First observe that when $X$ is distributed according to $F$, we have that 
	\[\frac{\EE(X|X>\Theta_n)}{\Theta_n}=1+\frac{1}{\Theta_n(1-F(\Theta_n))}\int_{\Theta_n}^{\infty}(1-F(s))\mathrm ds.\]
	Let $n_0$ be a sufficiently large positive integer such that $\Theta_n>0$ for every $n\ge n_0$.
	Then, for every $n\ge n_0$ and from the Von Mises representation we obtain that
	\begin{align*}
		\frac{1}{\Theta_n(1-F(\Theta_n))}\int_{\Theta_n}^{\infty}(1-F(s))\mathrm ds &= \int_{\Theta_n}^{\infty}\frac{\eta(s)}{\eta(\Theta_n)}\cdot \frac{1-V(s)}{\Theta_n(1-V(\Theta_n))}\mathrm ds \\
		&=\int_{1}^{\infty}\frac{\eta(\omega \Theta_n)}{\eta(\Theta_n)}\cdot \frac{1-V(\omega \Theta_n)}{1-V(\Theta_n)}\mathrm d\omega,
	\end{align*} 
	where the last equality comes by performing a change of variables.
	\begin{claim}\label{clm:tail-equiv-von-mises}
		For every $\omega\in (1,\infty)$ we have that $\displaystyle \frac{1-V(\omega\Theta_n)}{1-V(\Theta_n)}\to 0.$
	\end{claim}
    \begin{proof}
    Since $V$ is a Von Mises function, there exists an auxiliary function $\mu$ and a $z_0 \in \R$ such that 
	\[
	\frac{1-V(\omega\Theta_n)}{1-V(\Theta_n)}=
	\exp\left(- \int\limits_{z_0}^{\omega \Theta_n} \frac{1}{\aux(s)}\mathrm  ds  \right) 
	\exp\left(  \;\int\limits_{z_0}^{\Theta_n} \frac{1}{\aux(s)}\mathrm  ds \right) 
	=\exp\left(  -\int\limits_{\Theta_n}^{\omega \Theta_n} \frac{1}{\aux(s)}\mathrm  ds \right)
	\]
	To complete the proof it is sufficient to show that the integral above goes to $\infty$ when $n\to \infty$.
	By changing variables we have that
	\[
	\int\limits_{\Theta_n}^{\omega \Theta_n} \frac{1}{\aux(s)}\mathrm  ds = \int\limits_1^\omega \frac{\Theta_n}{\aux(y\Theta_n)} \mathrm d y.
	\]
	Since $\Theta_n\to \infty$, for every $y\in (1,\omega)$ we have that
	\begin{equation*}
		\lim_{n\to \infty}\frac{\Theta_n}{\aux(y \Theta_n)}=\frac{y \Theta_n}{\aux(y \Theta_n)}\cdot \frac{1}{y}= \infty,
	\end{equation*}
	where the limit holds thanks to Lemma~\ref{lem:von-mises-prop}. 
	We therefore conclude that  
	\begin{equation*}
		\lim_{n\to \infty}\frac{1-V(\omega\Theta_n)}{1-V(\Theta_n)}=\lim_{n\to \infty}\exp\left( -\int\limits_1^\omega \frac{\Theta_n}{\aux(y \Theta_n)} \mathrm d y \right)= 0.\qedhere
	\end{equation*}
    \end{proof}
	We conclude the proof of the lemma, using Claim~\ref{clm:tail-equiv-von-mises}.
	Since $\Theta_n\to \infty$, we have that $\eta(\Theta_n)\to \eta^{\star}$ and $\eta(\omega \Theta_n)\to \eta^{\star}$ for every $\omega\in (1,\infty)$.
	Therefore, we have that 
	\[\lim_{n\to \infty}\frac{\eta(\omega \Theta_n)}{\eta(\Theta_n)}\cdot \frac{1-V(\omega \Theta_n)}{1-V(\Theta_n)}= \frac{\eta^{\star}}{\eta^{\star}}\cdot 0=0,\]
	and thanks to the dominated convergence theorem this implies that
	\[
	\frac{1}{\Theta_n(1-F(\Theta_n))}\int_{\Theta_n}^{\infty}(1-F(s))\mathrm ds =\int_{1}^{\infty}\frac{\eta(\omega \Theta_n)}{\eta(\Theta_n)}\cdot \frac{1-V(\omega \Theta_n)}{1-V(\Theta_n)}\mathrm d\omega \to 0.\qedhere
	\]
\end{proof}

We are now ready to prove Theorem~\ref{thm:apx:multi}\ref{multiple-Gumbel} for the Gumbel family.

\begin{proof}[Proof of Theorem \ref{thm:apx:multi}\ref{multiple-Gumbel} for the Gumbel family]
Let $F$ be a distribution with extreme type in the Gumbel family such that $\omega_1(F)=\infty$. 
Consider a Von Mises pair $(V,\eta)$ that represents $F$ with $\lim_{t\to \infty}\eta(t)=\eta^{\star}>0$, guaranteed to exist by Lemma \ref{prop:Gumbel-prop}.
Let $\{a_n\}_{n\in\NN}$ and $\{b_n\}_{n\in\NN}$ be scaling and shifting sequences, respectively, for $V$.
For every positive integer $n$ consider $b_n^{\eta}=b_n+a_n\log \eta^{\star}$.
By Proposition~\ref{prop:numerator-multi}, we can lower bound the value of $\apx_k(F)$ by 
\begin{equation}\label{eq:gumbel-apx-program}
\sup_{U\in \RR}\liminf_{n\to \infty} \frac{\EE \left(X \vert X > a_nU+b_n^{\eta} \right)}{a_nU+b_n^{\eta}} \cdot \frac{a_nU+b_n^{\eta}}{\sum_{j=1}^k \EE(M_n^{j})}\cdot \sum_{j=1}^{k} \PP(M_n^{j}>a_nU+b_n^{\eta}).
\end{equation}
Our threshold $T_n$ will be defined as $T_n = a_n U + b_n^{\eta}$, where $U$ is the optimal solution to \eqref{eq:gumbel-apx-program}.
By Lemma~\ref{lem:mises-to-gumbel}\ref{limit-between}, we have  $a_nU+b_n^{\eta}\to \infty$ for every $U$ when $n\to \infty$, and therefore from Lemma~\ref{lem:gumbel-first-term} we obtain
\[\lim_{n\to \infty}\frac{\EE \left(X \vert X > a_nU+b_n^{\eta} \right)}{a_nU+b_n^{\eta}}=1
,\]
for every $U$.
Furthermore, Lemma~\ref{lem:mises-to-gumbel}\ref{lem:gumbel-k-second} implies that for every $U$ and every positive integer $k$ it holds $(a_nU+b_n^{\eta})/\sum_{j=1}^k \EE(M_n^{j})\to 1/k$.
We conclude that for every $U$ 
\[\lim_{n\to \infty} \frac{\EE \left(X \vert X > a_nU+b_n^{\eta} \right)}{a_nU+b_n^{\eta}} \cdot \frac{a_nU+b_n^{\eta}}{\sum_{j=1}^k \EE(M_n^{j})}=\frac{1}{k}.\]
By Lemma \ref{lem:mises-to-gumbel}\ref{sequences-gumbel}, $\{a_n\}_{n\in\NN}$ and $\{b_n^{\eta}\}_{n\in\NN}$ are scaling and shifting sequences, respectively, for $F$.
Therefore, by Theorem \ref{thm:leadbetter} we have 
\begin{align*}
\lim_{n\to \infty}\sum_{j=1}^{k} \PP(M_n^{j}>a_nU+b_n^{\eta})&= \lim_{n\to \infty}\sum_{j=1}^{k} \PP\left(\frac{M_n^{j}-b^{\eta}_n}{a_n}>U\right)\\
&=\sum_{j=1}^{k}\left(1-\exp\Big(- e^{-U}\Big)\sum_{s=0}^{j-1}\frac{ e^{-sU}}{s!}\right)\\
&=k-\exp\Big(-e^{-U}\Big)\sum_{j=1}^{k}\sum_{s=0}^{j-1}\frac{e^{-sU}}{s!},
\end{align*}
where the second equality follows from the fact that $F$ has extreme type Gumbel. Notice that the last term is non-negative for every $U$.
Furthermore, we get that 
\[\lim_{U\to \infty}\exp\Big(-e^{-U}\Big)\sum_{j=1}^{k}\sum_{s=0}^{j-1}\frac{ e^{-sU}}{s!}=\inf_{U\in \RR}\exp\Big(-e^{-U}\Big)\sum_{j=1}^{k}\sum_{s=0}^{j-1}\frac{ e^{-sU}}{s!}=0\]
since $\sum_{s=0}^{\infty}e^{-sU}/s!=\exp(-e^{-U})$.
We conclude that 
\[\sup_{U\in \RR}\lim_{n\to \infty} \frac{\EE \left(X \vert X > a_nU+b_n^{\eta} \right)}{a_nU+b_n^{\eta}} \cdot \frac{a_nU+b_n^{\eta}}{\sum_{j=1}^k \EE(M_n^{j})}\cdot \sum_{j=1}^{k} \PP(M_n^{j}>a_nU+b_n^{\eta})=\frac{1}{k}\cdot k=1,\]
and therefore $\apx_k(F)= 1$. 
That concludes the proof for the Gumbel family.
\end{proof}

\section{Tightness: Proof of Theorem \ref{thm:tightness}}
\label{app:tightness}

Before proving Theorem~\ref{thm:tightness} we state some technical results first.
In what follows we denote $\calP(y,k)=\sum_{j=0}^{k}e^{-y}y^{j}/j!$ the probability that that a random variable distributed according to a Poisson of mean $y$ is at most $k$.

\begin{lemma}
\label{lem:frechet-eq}
Let $F$ be the Pareto distribution with parameter $\alpha=2$. Then, for every positive integer $k$ we have $\apx_k(F)=\varphi_k(2)$.
\end{lemma}
\begin{proof}

Consider a positive integer $k$ and let $T_n$ be an optimal solution of the problem \footnote{ Note that $T_n$ could be $\infty$ for some $n$.} 
\[
 \sup_{T \in \RR_+} \EE \left(X \vert X > T\right)\frac{\sum_{j=1}^{k} \PP(M_n^{j}>T)}{\sum_{j=1}^k \EE(M_n^{j})},
 \]
 and let $\{a_n\}_{n\in \NN}$ be the scaling sequence for $F$, which by Lemma~\ref{lem:fre-key} is equal to $a_n=\sqrt{n}$. 
 We will show that the sequence of positive real values $\{T_n/a_n\}_{n\in \NN}$ is bounded. Suppose this is not the case, that is, the sequence $\{T_n/a_n\}_{n \in \NN}$ is unbounded. 
 From the optimallity of $T_n$ together with equation \eqref{eqn:comp:ratio2:k}, we have  
 $\apx_k(F)=\liminf_{n\to \infty} A_n\cdot B_n$
 where 
 \begin{align*}
     A_n&=\frac{\EE \left(X \vert X > T_n\right)}{T_n}\cdot \frac{1}{\sum_{j=1}^k \EE(M_n^j/a_n)},\quad B_n=\sum_{j=1}^{k} \frac{T_n}{a_n} (1-G_{n,j}(T_n)),
 \end{align*}
and $G_{n,j}$ is the distribution of the $j$-th order statistic $M_n^j$.
 Notice that the sequence $\{A_n\}_{n\in \NN}$ converges to  $\Gamma(k)/\Gamma(k+1-1/2)$ when $n\to \infty$, thanks to Lemma~\ref{lem:fre-integral} together with the argument we have already seen in the proof of Lemma~\ref{lem:frechet-technical-2}. 
 On the other hand, for each $j \in \{ 1, \dots, k\}$, it holds that $1-G_{n,j}(T_n)$ is at most $1-G_{n,1}(T_n)=1-(1-1/T_n^2)^n$, by the definition of the CDF of the Pareto distribution with parameter $\alpha = 2$, and therefore 
 \[
 B_n \le \frac{T_n}{a_n} k \left( 1-\left(1-\frac{1}{T_n^2}\right)^n \right) \le \frac{T_n}{a_n} k \cdot \frac{n}{T^2_n} = k \frac{a_n}{T_n} \to 0,
 \]
where the second inequality follows using the Bernoulli's inequality, the fact that $a_n=\sqrt{n}$, and the limit holds since we are assuming that the sequence ${T_n/a_n}_{n\in \NN}$ is unbounded. 

 Therefore, putting all together we obtain that if the sequence $\{T_n/a_n \}_{n \in \NN}$ is unbounded, we get $\apx_k(F)=0$ for all positive integer $k$, which contradicts the optimality of the sequence $\{T_n\}_{n \in \NN}$. Then, the sequence $\{T_n/a_n\}_{n \in N}$ is bounded and therefore there exist a convergent subsequence, namely $\{U_n\}_{n \in \NN}$, and  $U \in \RR_+$ such that $\lim_{n \to \infty} U_n=U$.
Applying Lemma~\ref{lem:frechet-technical-2} to the sequence $\{U_n\}_{n \in \NN}$  we obtain that 
\[\apx_k(F)= \frac{\Gamma(k)}{\Gamma(k+1-1/2)} U \exp(- U^{-2}) \sum_{j=1}^k \sum_{s=j}^{\infty} \frac{U^{-s 2}}{s!}.\]
It remains to show that 
$U$ is a maximizer in the optimization problem \eqref{eq:phi-function} defining the value $\varphi_k(2)$.
By contradiction, suppose that there exist $\widetilde{U} \neq U$ such that 
\begin{equation}
\label{ineq:main:thm}
  \widetilde{U} \exp(-\widetilde{U}^{-2}) \sum_{j=1}^k \sum_{s=j}^{\infty} \frac{\widetilde{U}^{-s 2}}{s!} >  U \exp(-U^{-2}) \sum_{j=1}^k \sum_{s=j}^{\infty} \frac{U^{-s 2}}{s!}.   
\end{equation}
Consider the sequence $\{\widetilde{T}_n \}_{n\in \NN}$, where $\widetilde{T}_n=a_n \widetilde{U}$ for every positive integer $n$. Note that $\widetilde{T}_n \to \infty$ and $\widetilde{T}_n/a_n \to \widetilde{U}$, and thus, due to Lemma~\ref{lem:frechet-technical-2} together with the inequality \eqref{ineq:main:thm},  we have 
 \[
\lim_{n\to \infty}\EE \left(X \vert X > T_n\right)\frac{\sum_{j=1}^{k} \PP(M_n^{j}>T_n)}{\sum_{j=1}^k \EE(M_n^{j})} < \lim_{n\to \infty}\EE \left(X \vert X > \widetilde{T}_n\right)\frac{\sum_{j=1}^{k} \PP(M_n^{j}>\widetilde{T}_n)}{\sum_{j=1}^k \EE(M_n^{j})}.\]

The last inequality implies that there exists a positive integer $n_0$ such that for all $n>n_0$, 
 \[ 
 \EE \left(X \vert X > T_n\right)\frac{\sum_{j=1}^{k} \PP(M_n^{j}>T_n)}{\sum_{j=1}^k \EE(M_n^{j})} <  \EE \left(X \vert X > \widetilde{T}_n\right)\frac{\sum_{j=1}^{k} \PP(M_n^{j}>\widetilde{T}_n)}{\sum_{j=1}^k \EE(M_n^{j})},
 \]
 which is a contradiction since $T_n$ was taken optimally for each positive integer $n$. 
 Therefore we conclude that 
 $U$ is a maximizer in the optimization problem \eqref{eq:phi-function} defining the value $\varphi_k(2)$,
obtaining that $\apx_k(F)=\varphi_k(2)$. 
\end{proof}
\begin{lemma}
\label{lem:fechet-technical-1}
For every positive integer $k$, there exists a unique optimal solution $x_k$ for the optimization problem (\ref{eq:phi-function}) when $\alpha=2$, and furthermore we have that $(k+1)^{-1/2}\le x_k\le k^{-1/2}$.
In particular, we have
\[\varphi_k(2)=\frac{\Gamma(k)}{\Gamma(k+1-1/2)} \left(x_k^{-1}\calP\left(x_k^{-2},k-1\right)+kx_k\left(1-\calP\left(x_k^{-2},k\right)\right)\right).\]
\end{lemma}
To prove the lemma we use the following two propositions. 
\begin{proposition}
	\label{lem:inequalities}	
	For every positive integer $k$ the following holds:
	\begin{enumerate}[label=\normalfont(\roman*)] 
	\item For every value $y\in (0,\infty)$ we have that $y\calP(y,k-1)\le k\calP(y,k)$. \label{poisson-ineq-0}
	\item $\calP(k,k)+\calP(k,k-1)>1$. \label{poisson-ineq-1}
	\item $k(1-\calP(k+1,k))-(k+1)\calP(k+1,k-1)>0$. \label{poisson-ineq-2}
	\end{enumerate}
\end{proposition}

\begin{proof}
    Recall that $\calP(y,k)=\sum_{j=0}^k\frac{e^{-y}y^j}{j!} = \frac{\Gamma(k+1, y)}{\Gamma(k+1)}$.
	For every $y\in (0,\infty)$, the inequality in \ref{poisson-ineq-0} comes directly from the following chain:
	\begin{align*}
    y \calP(y,k-1) = y \frac{\Gamma(k, y)}{\Gamma(k)} = y \frac{\int_y^{\infty} x^{k-1} e^{-x} \dif y}{\Gamma(k)} \leq \frac{\int_y^{\infty} x^k e^{-x} \dif y}{\Gamma(k)} = \frac{\Gamma(k+1, y)}{\Gamma(k)} = k \frac{\Gamma(k+1, y)}{\Gamma(k+1)} = k \calP(y,k)
	\end{align*}
	Next, we prove \ref{poisson-ineq-1}. 
	For every strictly positive integer $k$, let $\theta(k)$ be the value such that 
	\[\sum_{j=0}^{k-1}\frac{k^j}{j!}+\theta(k)\frac{k^k}{k!}=\frac{e^{k}}{2}.\]
	Therefore, we have that 
	\begin{align*}
	\calP(k,k)+\calP(k,k-1)&=\frac{e^{-k}k^k}{k!}+2\calP(k,k-1)\\
	                        &=\frac{e^{-k}k^k}{k!}+2\left(\frac{1}{2}-\theta(k)\frac{e^{-k}k^k}{k!}\right)=1+(1-2\theta(k))\frac{e^{-k}k^k}{k!}.
	\end{align*}
	Chow~\cite[Theorem 3]{choi1994medians} proved that for every $k$ we have that $1/3\le \theta(k)< 0.37$. 
	Therefore, we have that $\calP(k,k)+\calP(k,k-1)\ge 1+(1-2\cdot 0.37)\cdot e^{-k}k^k/k!>1$.
	
	Now we prove \ref{poisson-ineq-2}.
	By applying part \ref{poisson-ineq-0} with $y=k+1$ we get that $(k+1)\calP(k+1,k-1)\le k\calP(k+1,k)$.
	On the other hand, the median of a Poisson random variable of mean $k+1$ is at most $k+1$, and therefore $\calP(k+1,k)<1/2 \implies \calP(k+1,k)< 1-\calP(k+1,k)$.
	We conclude that $(k+1)\calP(k+1,k-1)\le k\calP(k+1,k)<k(1-\calP(k+1,k))$. 
\end{proof}

\begin{proposition}
\label{prop:unique-root}
Let $H$ and $L$ be two functions that are twice differentiable over $\RR_+$ and  satisfying the following conditions:
\begin{enumerate}[label=\normalfont(\alph*)] 
    \item $H(0)=L(0)$ and there exists $x\in (0,\infty)$ such that $H(x)=L(x)$.
    \item $H''(y)>L''(y)$ for every $y\in (0,\infty)$.
\end{enumerate}
Then, $x$ is the unique solution in $(0,\infty)$ of the equation $H(y)=L(y)$.
\end{proposition}

\begin{proof}
Suppose that there is a value $\tilde x\ne x$ such that $H(\tilde x)=L(\tilde x)$ and assume without loss of generality that $x<\tilde x$.
Consider the function $\phi(y)=H(y)-L(y)$ over $\RR_+$.
Since $\phi(0)=H(0)-L(0)=0$ and $\phi(x)=H(x)-L(x)=0$, the mean value theorem guarantees the existence of a point $x_1\in (0,x)$ such that $\phi'(x_1)=0$, that is $H'(x_1)=L'(x_1)$.
On the other hand, since $\phi(\tilde x)=H(\tilde x)-L(\tilde x)=0=\phi(x)$, the mean value theorem guarantees the existence of a point $x_2\in (x,\tilde x)$ such that $\phi'(x_2)=0$, that is, $H'(x_2)=L'(x_2)$.
Then, we have that $x_1<x_2$ and $\phi'(x_1)=\phi'(x_2)$, and therefore the mean value theorem once more guarantees the existence of a point $x_3\in (x_1,x_2)$ such that $\phi''(x_3)=0$, that is, $H''(x_3)=L''(x_3)$, but this contradicts the condition $(b)$ satisfied by $H$ and $L$. 
\end{proof}

\begin{proof}[Proof of Lemma \ref{lem:fechet-technical-1}]
	Consider the optimization problem in
		\eqref{eq:phi-function}, and for every positive integer $k$ consider its objective function when $\alpha=2$,
		\begin{equation*}
		    f_k(x)=x \exp(-x^{-\alpha}) \sum_{j=1}^k \sum_{s=j}^{\infty} \frac{x^{-2s }}{s!}.
		\end{equation*}
	The derivative of $f_k$ is given by
	\[f'_k(y)=\exp(-y^{-2})\left(\left(1+\frac{2}{y^2}\right)\sum_{j=1}^k \sum_{s=j}^{\infty} \frac{y^{-2 s }}{s!}-2 \sum_{j=1}^k \sum_{s=j}^{\infty} \frac{y^{-2 s }}{(s-1)!}\right).\]
	\noindent By expanding the last summation we have that 
	\[
	2 \sum_{j=1}^k \sum_{s=j}^{\infty} \frac{y^{-2 s }}{(s-1)!}=
	\frac{2}{y^{2}}\sum_{j=1}^k \sum_{s=j}^{\infty} \frac{y^{-2 s }}{s!}+\frac{2}{y^{2}}\sum_{j=0}^{k-1}\frac{y^{-2 j}}{j!},
	\]
	and therefore we recover that 
	\begin{align*}
	f'_k(y)&=\exp(-y^{-2})\left(\left(1+\frac{2}{y^2}\right)\sum_{j=1}^k \sum_{s=j}^{\infty} \frac{y^{-2 s }}{s!}-\frac{2}{y^{2}}\sum_{j=1}^k \sum_{s=j}^{\infty} \frac{y^{-2 s }}{s!}-\frac{2}{y^{2}}\sum_{j=0}^{k-1}\frac{y^{-2 j}}{j!}\right)\\
	&=\exp(-y^{-2})\left(\sum_{s=1}^{\infty} \sum_{j=1}^{\min(k,s)} \frac{y^{-2 s }}{s!}-\frac{2}{y^{2}}\sum_{j=0}^{k-1}\frac{y^{-2 j}}{j!}\right)\\
	&=\exp(-y^{-2})\left(\sum_{s=0}^{k-1} \frac{y^{-2(s+1) }}{s!}+k\sum_{s=k+1}^{\infty} \frac{y^{-2 s }}{s!}-\frac{2}{y^{2}}\sum_{j=0}^{k-1}\frac{y^{-2 j}}{j!}\right)\\
	&=\exp(-y^{-2})\left(k\sum_{s=k+1}^{\infty} \frac{y^{-2 s }}{s!}-\frac{1}{y^{2}}\sum_{j=0}^{k-1}\frac{y^{-2 j}}{j!}\right)\\
	&=k(1-\calP(y^{-2},k))-y^{-2}\calP(y^{-2},k-1)
	\end{align*}
	In what follows we prove the existence of a root for the derivative.
	Consider $\bar y=k^{-1/2}$.
	Then, we have that 
    $f'_k(\bar y)= k\left(1-\calP\left(k,k\right)-\calP\left(k,k-1\right)\right)< 0$, 
	where the last inequality holds by Proposition~\ref{lem:inequalities}\ref{poisson-ineq-1}. 
	Consider now $\tilde y$ such that $\bar y^{-2}=k+1$.
	Then, we have that
	$f'_k(\tilde y)=k(1-\calP(k+1,k))-(k+1)\calP(k+1,k-1)>0$,
	where the last inequality holds by Proposition \ref{lem:inequalities}\ref{poisson-ineq-2}.
	Therefore, the function $f'_k$ changes sign in the interval $[(k+1)^{-1/2},k^{-1/2}]$ and by the continuity of $f_k'$ we conclude that there exists $x_k\in [(k+1)^{-1/2},k^{-1/2}]$ such that $f_k'(x_k)=0$.
	We now show that in fact $x_k$ is the unique root of $f'_k$.
	Since $y^{-2}$ is a bijective function from $(0,\infty)$ to $(0,\infty)$, proving that $f_k'$ has a single root is equivalent to show that the function defined by $G_k(y)=f'_k(y^{-1/2})\cdot \exp(y)$ has a unique root.
	First observe that for every $y\in \RR_+$ we have that 
	\begin{align*}
	G_k(y)&=ke^y-ke^y\calP(y,k)-ye^y\calP(y,k-1)\\
	      &=k(e^y-1)-k\sum_{j=1}^k\frac{y^j}{j!}-\sum_{j=0}^{k-1}\frac{y^{j+1}}{j!}=k(e^y-1)-\sum_{j=0}^{k-1}\frac{y^{j+1}}{j!}\left(\frac{k}{j+1}-1\right).
	\end{align*}
	Consider the functions over $\RR_+$ defined by $H_k(y)=k(e^y-1)$ and $L_k(y)=G_k(y)-H_k(y)$.
	Observe that $H_k(0)=L_k(0)=0$ and $H_k(x_k^{-2})=L(x_k^{-2})$, since $H_k(x_k^{-2})-L_k(x_k^{-2})=G_k(x_k^{-2})=f_k'(x_k)\exp(x_k^{-2})=0$.
	Furthermore, we claim that $H_k''(y)> L_k''(y)$ for every $y\in \RR_+$.
	By Proposition \ref{prop:unique-root} we conclude that $x_k^{-2}$ is the unique solution in $(0,\infty)$ of the equation $H_k(y)=L_k(y)$, which is equivalent to $x_k$ being the unique solution in $(0,\infty)$ of $f'_k(y)=0$.
	Observe that for every $y\in \RR_+$ we have that 
	\begin{align*}
	L_k(y)&=\sum_{j=0}^{k-1}\frac{y^{j+1}}{j!}\left(\frac{k}{j+1}-1\right)=y(k-1)+\sum_{j=0}^{k-2}\frac{y^{j+2}}{(j+1)!}\left(\frac{k}{j+2}-1\right).    
	\end{align*}
	Therefore, for every $y\in \RR_+$ we have that 
	\begin{align*}
	L_k''(y)&=\sum_{j=0}^{k-2}\frac{y^{j}(j+2)(j+1)}{(j+1)!}\left(\frac{k}{j+2}-1\right)\\
	&=\sum_{j=0}^{k-2}\frac{y^{j}}{j!}\left(k-(j+2)\right)\le k\sum_{j=0}^{k-2}\frac{y^{j}}{j!}<ke^y=H_k''(y),    
	\end{align*}
	where the strict inequality holds by the Taylor expansion of $e^y$.

	To conclude that $x_k$ maximizes $f_k$ it is sufficient to show that the second derivative of $f_k$ in $x_k$ is negative.
	Observe that this is sufficient since $f_k$ is non-negative over $\RR_+$ and $\lim_{y\to 0}f_k(y)=\lim_{y\to \infty}f_k(y)=0$.
	Consider the function $b_k(y)=f_k'(y)\exp(y^{-2})$.
	In particular, we have that 
	$$b_k(y)=k\sum_{s=k+1}^{\infty} \frac{y^{-2 s }}{s!}-\frac{1}{y^{2}}\sum_{j=0}^{k-1}\frac{y^{-2 j}}{j!}$$
	for every $y\in \RR_+$.
	Furthermore, we have that $b_k'(x_k)=\exp(x_k^{-2})f_k''(x_k)-2x_k^{-3}\exp(x_k^{-2})f'_k(x_k)=\exp(x_k^{-2})f''_k(x_k)$, since $f'_k(x_k)=0$.
	Therefore, it remains to check that $b_k'(x_k)<0$ since this implies that $f_k''(x_k)<0$ by the positivity of the exponential $\exp(x_k^{-2})$.
	Rearranging terms, we get that 
	\begin{align*}
		b_k'(x_k)&=-2k\sum_{j=k+1}^{\infty}\frac{x_k^{-2j-1}}{(j-1)!}+2\sum_{j=0}^{k-1}\frac{x_k^{-2j-3}}{j!}(j+1)\\
		&=-2kx_k^{-3}\sum_{j=k}^{\infty}\frac{x_k^{-2j}}{j!}+2\sum_{j=1}^{k-1}\frac{x_k^{-2j-3}}{(j-1)!}+2\sum_{j=0}^{k-1}\frac{x_k^{-2j-3}}{j!}\\
		&=-2kx_k^{-3}\sum_{j=k}^{\infty}\frac{x_k^{-2j}}{j!}+2x_k^{-5}\sum_{j=0}^{k-2}\frac{x_k^{-2j}}{j!}+2x_k^{-3}\sum_{j=0}^{k-1}\frac{x_k^{-2j}}{j!}\\
		&=-2kx_k^{-3}\sum_{j=k}^{\infty}\frac{x_k^{-2j}}{j!}+(2x_k^{-3}+2x_k^{-1})x_k^{-2}\sum_{j=0}^{k-1}\frac{x_k^{-2j}}{j!}-2\frac{x_k^{-2k-3}}{(k-1)!}\\
		&=-2kx_k^{-3}\sum_{j=k}^{\infty}\frac{x_k^{-2j}}{j!}+(2x_k^{-3}+2x_k^{-1})k\sum_{j=k+1}^{\infty}\frac{x_k^{-2j}}{j!}-2\frac{x_k^{-2k-3}}{(k-1)!},
	\end{align*}
	where the last equality holds since $x_k$ satisfies that $f'_k(x_k)=0$.
	Therefore, we get that
	\begin{align*}
		x_kb_k'(x_k)&=-2k\frac{x_k^{-2k-2}}{k!}+2k\sum_{j=k+1}^{\infty}\frac{x_k^{-2j}}{j!}-2\frac{x_k^{-2k-2}}{(k-1)!}\\
		&=-4k\frac{x_k^{-2k-2}}{k!}+2x_k^{-2}\sum_{j=0}^{k-1}\frac{x_k^{-2j}}{j!}\le 2x_k^{-2}\left(-2\frac{k^{k+1}}{k!}+\sum_{j=0}^{k-1}\frac{(k+1)^{j}}{j!}\right),
	\end{align*}
	the second equality holds by using that $f'_k(x_k)=0$ and the inequality holds since $k\le x_k^{-2}\le k+1$ for every positive integer $k$.
	To finish we have for every positive integer $k$ it holds that
	\[\sum_{j=0}^{k-1}\frac{(k+1)^{j}}{j!}\le \frac{1}{2}e^{k+1}-\frac{(k+1)^k}{k!}<2\frac{k^{k+1}}{k!}.\]
	From this we get directly that $b_k'(x_k)$ and therefore $f_k''(x_k)<0$ for every positive integer $k$.   
\end{proof}

\begin{proof}[Proof of Theorem \ref{thm:tightness}]
We first observe that thanks to Lemma~\ref{lem:frechet-eq}, showing Theorem~\ref{thm:tightness}  is equivalent to prove that if $F$ is the Pareto distribution with parameter $\alpha=2$, then  for every $\varepsilon>0$ there exists a positive integer $k_{\varepsilon}$ such that for every $k\ge k_{\varepsilon}$ it holds that $\varphi_k(2)\le 1-(1-\varepsilon)/\sqrt{2\pi k}$. In what follows, we prove the latter statement. 
Notice that thanks to Lemma \ref{lem:fechet-technical-1} we have that there exists an optimal solution $x_k\in [(k+1)^{-1/2},k^{-1/2}]$ such that
	\begin{align*}
		\varphi_k(2)&=\frac{\Gamma(k)}{\Gamma(k+1/2)}\left(\frac{1}{x_k}\calP_{k-1}(x_k^{-2})+kx_k(1-\calP_k(x_k^{-2}))\right)\\
		&=\frac{\Gamma(k)\sqrt{k-1}}{\Gamma(k+1/2)}\left(\frac{1}{x_k\sqrt{k-1}}\calP_{k-1}(x_k^{-2})+\frac{kx_k}{\sqrt{k-1}}(1-\calP_k(x_k^{-2}))\right).
	\end{align*}
	\begin{claim}
	\label{claim:gamma-upper}
	For every positive integer $k$ we have that $\displaystyle \frac{\Gamma(k)\sqrt{k-1}}{\Gamma(k+1/2)}\le 1$.
	\end{claim}
    \begin{proof}
    Thanks to the Gautschi inequality we have that $\frac{\Gamma(k+1)}{\Gamma(k+1/2)}\le \sqrt{k+1}$.
	Therefore, we have that 
	\[\frac{\Gamma(k)\sqrt{k-1}}{\Gamma(k+1/2)}=\frac{\Gamma(k+1)}{\Gamma(k+1/2)}\cdot \frac{\sqrt{k-1}}{k}\le \sqrt{k+1}\cdot \frac{\sqrt{k-1}}{k}=\frac{\sqrt{k^2-1}}{k}<1,\]
	which concludes the proof. 
    \end{proof}
	\noindent By using Claim \ref{claim:gamma-upper} we have that for every positive integer $k\ge 2$
	\begin{align*}
		\varphi_k(2)&\le\frac{1}{x_k\sqrt{k-1}}\calP_{k-1}(x_k^{-2})+\frac{kx_k}{\sqrt{k-1}}(1-\calP_k(x_k^{-2}))\\
		&\le \frac{\sqrt{k+1}}{\sqrt{k-1}}\calP_{k-1}(x_k^{-2})+\frac{\sqrt{k}}{\sqrt{k-1}}(1-\calP_k(x_k^{-2}))\\
		&\le \frac{\sqrt{k+1}}{\sqrt{k-1}}\Big(\calP_{k-1}(x_k^{-2})+1-\calP_k(x_k^{-2})\Big),
	\end{align*}
	where the second inequality follows thanks to Lemma \ref{lem:fechet-technical-1}.
	We recall that $\calP_k(x_k^{-2})$ is the probability that a Poisson random variable of mean $x_k^{-2}$ is at most $k$, and therefore we have that 
	\[\calP_{k-1}(x_k^{-2})+1-\calP_k(x_k^{-2})=1-\frac{e^{-x_k^{-2}}x_k^{-2k}}{k!},\]
	and this equality together with the above inequality imples that 
	\[\varphi_k(2)\le \left(\frac{k+1}{k-1}\right)^{1/2}\left(1-\frac{e^{-x_k^{-2}}x_k^{-2k}}{k!}\right).\]
	Then, we get that 
	\begin{equation*}
		(1-\varphi_k(2))\sqrt{2 k \pi}\ge \left(1-\left(\frac{k+1}{k-1}\right)^{1/2}\right)\sqrt{2 k \pi} +\left(\frac{k+1}{k-1}\right)^{1/2}\frac{e^{-x_k^{-2}}x_k^{-2k}}{k!}\sqrt{2 k \pi}.
	\end{equation*}
	Note that the function $e^{-y}y^k$ is decreasing in the interval $[k,\infty)$ and since $k\le x_k^{-2}\le k+1$ we have that $e^{-x_k^{-2}}x_k^{-2k}\ge e^{-(k+1)}(k+1)^k$.
	Therefore, we have that 
	\begin{align*}
		&(1-\varphi_k(2))\sqrt{2 k \pi}\\
		&\ge \left(1-\left(\frac{k+1}{k-1}\right)^{1/2}\right)\sqrt{2 k \pi} +\left(\frac{k+1}{k-1}\right)^{1/2}\frac{e^{-(k+1)}(k+1)^k}{k!}\sqrt{2 k \pi}\\
		&= \left(1-\left(\frac{k+1}{k-1}\right)^{1/2}\right)\sqrt{2 k \pi} +\left(\frac{k+1}{k-1}\right)^{1/2}\cdot \frac{e^{-k}k^k}{k!}\sqrt{2 k \pi}\cdot e^{-1}\cdot \left(1+\frac{1}{k}\right)^k.
	\end{align*}
	We conclude by observing that the following limits hold:
	\begin{align*}\lim_{k\to \infty}\left(1-\left(\frac{k+1}{k-1}\right)^{1/2}\right)\sqrt{2 k \pi}=0,&\quad \lim_{k\to \infty}\frac{e^{-k}k^k}{k!}\sqrt{2 k \pi}=1,\\
	\lim_{k\to \infty}\left(\frac{k+1}{k-1}\right)^{1/2}=1,&\quad\lim_{k\to \infty}\left(1+\frac{1}{k}\right)^k=e,
	\end{align*}
	where the second limit is given by the Stirling approximation.
\end{proof}

\section{Revenue Guarantees: Proof of Theorem \ref{thm:stability}}
\label{app:VVEV}

In the following, we prove Theorem~\ref{thm:stability}.
Recall that that a pair $(V,\eta)$ {\it smoothly represents} a distribution $F$ if $(V,\eta)$ is a Von Mises representation of $F$ and $\lim_{t \to \infty} \eta'(t)=0$ (see Lemma \ref{prop:Gumbel-prop}).
We first state the following simple lemma.
\begin{lemma}
\label{lem:nice-virtual}
Let $F$ be a distribution satisfying the extreme value condition. 
\begin{enumerate}[label=\normalfont(\roman*)]
	\item If $F$ has an extreme type Fr\'echet of parameter $\alpha\in (1,\infty)$ and it satisfies the asymptotic regularity condition, then $\lim_{t\to \infty}\phi_F(t)/t=1-1/\alpha.$\label{regular-a}
	\item If $F$ is a Von Mises function with auxiliary function $\aux$, then for every $t\in (\omega_0(F),\omega_1(F))$ we have that $\phi_F(t)=t-\aux(t)$.\label{regular-b} 
\end{enumerate}
\end{lemma}

\begin{proof}
If $F$ satisfies the asymptotic regularity condition it follows directly that
$\lim_{t\to \infty}\phi_F(t)/t=1-\lim_{t\to \infty}(1-F(t))/(tf(t))=1-1/\alpha.$
That proves \ref{regular-a}. 
Let $F$ be a Von Mises function with auxiliary function $\aux$ and $t\in (\omega_0(F),\omega_1(F))$. 
By the Von Mises representation~(\ref{V-M:repr}) we have that 
\[f(t)=\theta\exp\left( - \int\limits_{\omega_0(F)}^x \frac{1}{\aux(u)} \mathrm du  \right)\cdot \frac{1}{\aux(t)}=\frac{1-F(t)}{\aux(t)},\]
and therefore $\phi_F(t)=t-\aux(t)$.
That concludes \ref{regular-b}.
\end{proof}

%

We say that two distributions $F$ and $G$ are {\it tail equivalent} if $\omega_1(F)=\omega_1(G)=\omega_1$ and there is a $\beta>0$ such that
$\lim_{x\to \omega_1}(1-G(x))/(1-F(x))=\beta.$
We call the value $\beta$ the {\it asymptotic tail ratio} between $F$ and $G$.
We state the following result from the extreme value theory literature.
\begin{theorem}[see, e.g., \cite{libroextremo}]
\label{thm:tail-behavior}
Let $F$ and $G$ be two tail equivalent distributions with asymptotic tail ratio $\beta$, and suppose that $F$ satisfies the extreme value condition.
Then, the following holds:
\begin{enumerate}[label=\normalfont(\alph*)] 
	\item If $F$ has extreme type Fr\'echet of parameter $\alpha\in (1,\infty)$, we have that $G$ has an extreme type distribution equal to $\Phi_{\alpha}(\beta^{-1/\alpha} t)$, where $\Phi_{\alpha}$ is the Fr\'echet distribution of parameter $\alpha$.\label{thm:tail-behavior-frechet}
	\item If $F$ has extreme type Gumbel, we have that $G$ has an extreme type distribution equal to $\Lambda(t-\ln(\beta))$, where $\Lambda$ is the Gumbel distribution. \label{thm:tail-behavior-gumbel}
\end{enumerate}
\end{theorem}
The following lemma, together with Theorem~\ref{thm:tail-behavior} are the main tools to prove Theorem~\ref{thm:stability}.

\begin{lemma}
\label{lem:virtual-tail}
Let $F$ be a distribution that satisfies the extreme value condition. 
Then, the following holds.
\begin{enumerate}[label=\normalfont(\alph*)] 
	\item If $F$ has extreme type Fr\'echet of parameter $\alpha\in (1,\infty)$, and if its satisfies the asymptotic regularity condition, then $F$ and $F_{\phi}$ are tail equivalent with asymptotic tail ratio 
\begin{equation*}
\beta_{\alpha}=\lim_{t\to \infty}\frac{1-F_{\phi}(t)}{1-F(t)}=\left(1-\frac{1}{\alpha}\right)^{\alpha}.
\end{equation*}\label{virtual-tail-a}
	\item If $F$ has extreme type Gumbel, then $F$ and $F_{\phi}$ are tail equivalent with asymptotic tail ratio 
\begin{equation*}
\beta=\lim_{t\to \infty}\frac{1-F_{\phi}(t)}{1-V(t)}=\frac{\eta^{\star}}{e},
\end{equation*}
where $(V,\eta)$ is a Von Mises representation of $F$ according to \eqref{V-M:repr}.\label{virtual-tail-b}
\end{enumerate}
\end{lemma}
Observe that the asymptotic tail ratio for distributions with extreme type in the Fr\'echet family satisfies $\lim_{\alpha\to \infty}\beta_{\alpha}=1/e$.
Furthermore, if $F$ is a Von Mises function then the asymptotic tail ratio is equal to $1/e$ since we can always take the representation $(F,\mathsf{1}_{\RR})$.
We first show that the asymptotic tail ratio is preserved for Von Mises functions, and then we show how to extend for the Gumbel case.  
In the following we conclude Theorem~\ref{thm:stability} using Lemma~\ref{lem:virtual-tail}.

\begin{proof}[Proof of Theorem~\ref{thm:stability}]
Let $F$ be a distribution with extreme type Fr\'echet of parameter $\alpha\in (1,\infty)$.
By Lemma~\ref{lem:virtual-tail}, the distributions $F$ and $F_{\phi}$ are tail equivalent with asymptotic tail ratio equal to $\beta=((\alpha-1)/\alpha)^{\alpha}$.
By Theorem~\ref{thm:tail-behavior}\ref{thm:tail-behavior-frechet} we have that $F_{\phi}(t)=\Phi_{\alpha}(\beta^{-1/\alpha}t)=\Phi_{\alpha}(\alpha t/(\alpha-1))$. 
Now let $F$ be a distribution with extreme type Gumbel, and let $(V,\eta)$ a Von Mises representation according to Lemma~\ref{prop:Gumbel-prop}.
Thanks to Lemma~\ref{lem:virtual-tail}, we have that
\begin{align*}
    \lim_{t\to \infty}\frac{1-F_{\phi}(t)}{1-F(t)}&=\lim_{t\to \infty}\frac{1-F_{\phi}(t)}{1-V(t)}\cdot \frac{1-V(t)}{1-F(t)}=\lim_{t\to \infty}\frac{1-F_{\phi}(t)}{1-V(t)}\cdot \frac{1}{\eta(t)}=\frac{\eta^{\star}}{e}\cdot \frac{1}{\eta^{\star}}=\frac{1}{e},
\end{align*}
where the limits follows thanks to Lemma~\ref{lem:virtual-tail} and since $\lim_{t\to \infty}\eta(t)=\eta$.
Therefore, according to Theorem~\ref{thm:tail-behavior}\ref{thm:tail-behavior-gumbel} we have that that the virtual valuation has an extreme type with distribution $F_{\phi}(t)=\Lambda(t-\ln(1/e))=\Lambda(t+1)$.
To conclude the theorem, observe that for every $t\ge 0$ we have that $\PP(\phi_F(v)>t)=\PP(\phi_F^+(v)>t)$ and therefore $F_{\phi}$ and $F_{\phi}^+$ are tail equivalent. 
\end{proof}

Before proving Lemma~\ref{lem:virtual-tail} we need a few properties regarding the Von Mises representation and the corresponding auxiliary function. 

\begin{proposition}
\label{prop:asympt}
Let $F$ be a distribution with extreme type Gumbel that can be smoothly represented by a pair $(V,\eta)$ with auxiliary function $\aux:(z_0,\infty)\to \RR_+$.
Then, the following holds.
\begin{enumerate}[label=\normalfont(\roman*)] 
    \item For every $\delta>0$ there exists $t_{\delta}\in \RR_+$ such that for every $u>t\ge t_{\delta}$ we have $|\aux(u)-\aux(t)|\le \delta(u-t)$.\label{prop4a}
    \item For every $\delta>0$ there exists $t_{\delta}\in \RR_+$ such that for every $v\ge t_{\delta}$ we have that $\left|\frac{1-F(v)}{f(v)\aux(v)}-1\right|\le \delta$.\label{prop4b}
\end{enumerate}
\end{proposition}

\begin{proof}[Proof of Proposition~\ref{prop:asympt}]
Since $\lim_{t\to \infty}\aux'(t)=0$, there exists $t_{\delta}\in \RR_+$ such that for every $t\ge t_{\delta}$ we have $|\aux'(t)|\le \delta$.
In particular, we have that for every $u>t\ge t_{\delta}$, 
\begin{align*}
\aux(u)&=\aux(t)+\int_t^u \aux'(s)\mathrm ds\le \aux(t)+\delta(u-t),\\
\aux(u)&=\aux(t)+\int_t^u \aux'(s)\mathrm ds\ge \aux(t)-\delta(u-t),
\end{align*}
which concludes part \ref{prop4a}.
For part \ref{prop4b}, observe first from the Von Mises representation $(V,\eta)$ we have that there exists $t_0\in \RR_+$ such that for every $t\ge t_0$ we have $\aux(t)=(1-V(t))/V'(t)$. Since $1 - F(t) = \eta(t) \left(1 - V(t)\right)$ by Lemma~\ref{prop:Gumbel-prop}, we have that
\[
\frac{f(t)}{V'(t)} = \eta(t) - \eta'(t)\frac{1 - V(t)}{V'(t)} = \eta(t) - \eta'(t) \aux(t),
\]
and thus
\begin{equation*}
    \frac{1-F(t)}{f(t)\aux(t)}=\eta(t)\frac{V'(t)}{f(t)}=\frac{\eta(t)}{\eta(t)-\eta'(t)\aux(t)}.
\end{equation*}
Since $\lim_{t \to \omega_1(F)} \eta'(t)=0$ there exists $t_1\in \RR_+$ such that for every $t\ge t_1$, we have that $|\eta'(t)|\le 1/t$.
Therefore, for every $t\ge \max\{t_0,t_1\}$ we have that
\begin{equation*}
  \frac{\eta(t)}{\eta(t)+\frac{\aux(t)}{t}}\le  \frac{\eta(t)}{\eta(t)-\eta'(t)\aux(t)}\le \frac{\eta(t)}{\eta(t)-\frac{\aux(t)}{t}}.
\end{equation*}
Since $\lim_{t\to \infty}\aux(t)/t=0$ and $\eta^\star = \lim_{t \to \infty} \eta(t)$, we conclude the result.
\end{proof}

\begin{proof}[Proof of Lemma~\ref{lem:virtual-tail}]
Consider first a valuation $v$ with extreme type Fr\'echet of parameter $\alpha$, with $\alpha\in (1,\infty)$, and satisfying the asymptotic regularity condition. 
Thanks to Lemma~\ref{lem:nice-virtual}\ref{regular-a}, for every $\varepsilon>0$ there exists $t_0\in \RR$ such that for every $x\ge t_0$ we have that
\begin{equation}
\label{eq:asympt}
\frac{\alpha-1}{\alpha}(1-\varepsilon)x\le \phi_F(x)\le \frac{\alpha-1}{\alpha}(1+\varepsilon)x.
\end{equation}
Therefore for every $t\ge t_0$, if we have $\phi_F(v)>t$ then $v>\frac{\alpha (1+\varepsilon)^{-1}t}{(\alpha-1)}$.
It follows on the one hand that
\begin{equation}
\label{eq:asym-ub}
1-F_{\phi}(t)=\PP(\phi_F(v)>t)\le \PP\left(v>\frac{\alpha (1+\varepsilon)^{-1}t}{(\alpha-1)}\right),
\end{equation}
and on the other hand, if $v>\frac{\alpha (1-\varepsilon)^{-1}t}{(\alpha-1)}$ then from (\ref{eq:asympt}) we have that $\phi_F(v)>t$.
Thus,
\begin{equation}
\label{eq:asym-lb}
1-F_{\phi}(t)=\PP(\phi_F(v)>t)\ge \PP\left(v>\frac{\alpha (1-\varepsilon)^{-1}t}{(\alpha-1)}\right).
\end{equation}
By Lemma~\ref{lem:fre-key}, there exists a slowly varying function $\ell_F$ such that $1-F(s)=s^{-\alpha}\ell_F(s)$ for every $s\in \RR_+$.
In particular, we have
\begin{align*}
\PP\left(v>\frac{\alpha (1+\varepsilon)^{-1}t}{(\alpha-1)}\right)&=t^{-\alpha}\left(1-\frac{1}{\alpha}\right)^{\alpha}\left(1+\varepsilon\right)^{\alpha}\ell_{F}\left(\frac{\alpha (1+\varepsilon)^{-1}t}{(\alpha-1)}\right)\\									&=(1-F(t))\left(1-\frac{1}{\alpha}\right)^{\alpha}\left(1+\varepsilon\right)^{\alpha}\ell_{F}\left(\frac{\alpha (1+\varepsilon)^{-1}t}{(\alpha-1)}\right)\frac{1}{\ell_F(t)},
\end{align*}
and furthermore we get
\begin{align*}
\PP\left(v>\frac{\alpha (1-\varepsilon)^{-1}t}{(\alpha-1)}\right)&=t^{-\alpha}\left(1-\frac{1}{\alpha}\right)^{\alpha}\left(1-\varepsilon\right)^{\alpha}\ell_{F}\left(\frac{\alpha (1-\varepsilon)^{-1}t}{(\alpha-1)}\right)\\
	&=(1-F(t))\left(1-\frac{1}{\alpha}\right)^{\alpha}\left(1-\varepsilon\right)^{\alpha}\ell_{F}\left(\frac{\alpha (1-\varepsilon)^{-1}t}{(\alpha-1)}\right)\frac{1}{\ell_F(t)}.
\end{align*}
Since $\ell_F$ is slowly varying, there exists $t_1\in \RR_+$ such that for every $t\ge t_1$ we have that 
\begin{equation*}
1-\varepsilon\le \ell_{F}\left(\frac{\alpha (1+\varepsilon)^{-1}t}{(\alpha-1)}\right)\frac{1}{\ell_F(t)}\le 1+\varepsilon,
\end{equation*}
\begin{equation*}
1-\varepsilon\le \ell_{F}\left(\frac{\alpha (1-\varepsilon)^{-1}t}{(\alpha-1)}\right)\frac{1}{\ell_F(t)}\le 1+\varepsilon.
\end{equation*}
Then, from inequalities (\ref{eq:asym-ub})-(\ref{eq:asym-lb}) together with the inequalities above we get that for every $t\ge \max\{t_0,t_1\}$, we have
\begin{equation*}
\left(1-\frac{1}{\alpha}\right)^{\alpha}(1-\varepsilon)^{\alpha+1}\le \frac{1-F_{\phi}(t)}{1-F(t)}\le \left(1-\frac{1}{\alpha}\right)^{\alpha}(1+\varepsilon)^{\alpha+1},
\end{equation*}
from where we obtain that 
$\lim_{t\to \infty}(1-F_{\phi}(t))/(1-F(t))=(1-1/\alpha)^{\alpha}.$
That concludes part \ref{virtual-tail-a}.  

Now let $F$ be a distribution with extreme type Gumbel and consider a Von Mises representation $(V,\eta)$, with auxiliary function $\aux:(z_0,\infty)\to \RR_+$ and $z_0>\omega_0(F)$.
By Lemma~\ref{lem:nice-virtual}\ref{regular-b}, for every $t\in (z_0,\infty)$ we have that $\phi_V(t)=t-\aux(t)$ and then $1-V_{\phi}(t)=\PP(\phi_V(v)>t)=\PP(v-\aux(v)>t)$.
\noindent Thus, using Proposition~\ref{prop:asympt} we have that for every $\varepsilon\in (0,1/3)$ there exists $t_0\in \RR_+$ such that for every $v\ge t\ge t_0$ we have that $(1-F(v))/f(v)\ge (1-\varepsilon)\aux(v)$ and $\aux(v)\ge \aux(t)+\varepsilon(v-t).$
Therefore, if $\phi_F(v)>t$ we have 
\begin{align*}
v>\frac{1-F(v)}{f(v)}+t&\ge (1-\varepsilon)\aux(v)+t\\
            &\ge (1-\varepsilon)(\aux(t)-\varepsilon(v-t))+t\\
            &=(1-\varepsilon)\aux(t)-(\varepsilon-\varepsilon^2) v+(1-\varepsilon^2+\varepsilon)t\\
            &=(1-\varepsilon)\aux(t)+t-(\varepsilon-\varepsilon^2) (v-t) \hfill \iff \\
(v-t)(1+\varepsilon-\varepsilon^2) &\geq (1-\varepsilon)\aux(t) \hfill \iff \\
v &\geq t + \frac{1-\varepsilon}{1+\varepsilon-\varepsilon^2} \aux(t),
\end{align*}
where the third inequality holds by the positivity of $\mu$ and the fact that $v>(1-\varepsilon)\aux(v)+t\ge t\ge t_0$. Let $\xi(\varepsilon)=(1-\varepsilon)/(1-\varepsilon^2+\varepsilon)$. Then, from the inequality above we get
\begin{align}
\label{eq:lower-von}
\PP(\phi_F(v)>t)&\le \PP\left(v>t+\xi(\varepsilon)\aux(t)\right)=1-F\left(t+\xi(\varepsilon)\aux(t)\right).
\end{align}
On the other hand, by Proposition~\ref{prop:asympt} there exists $t_1\in \RR_+$ such that for every $v\ge t\ge t_1$ we have $(1-F(v))/f(v)\le (1+\varepsilon)\aux(v)$ and $\aux(t)\ge \aux(v)-\varepsilon^2(v-t)$. 
Then, for every $t\ge \max\{t_0,t_1\}$, if $v>t+(1+2\varepsilon)\aux(t)$ and $v\le t+\frac{1}{\varepsilon(1+\varepsilon)}\aux(t)$ we have
\begin{align*}
v&> t+\varepsilon \aux(t)+(1+\varepsilon)\aux(t)\\
			&\ge t+\varepsilon \aux(t)+(1+\varepsilon)(\aux(v)-\varepsilon^2(v-t))\\
			&= t+(1+\varepsilon)\aux(v)+(1+\varepsilon)\varepsilon^2\left(t+\frac{1}{\varepsilon(1+\varepsilon)}\aux(t)-v\right)\ge t+\frac{1-F(v)}{f(v)}.
\end{align*}
Therefore, for every $t\ge \max\{t_0,t_1\}$ we have
\begin{align*}
\PP(\phi_F(v)>t)&\ge \PP\left(t+\frac{1}{\varepsilon(1+\varepsilon)}\aux(t)\ge v>t+(1+2\varepsilon)\aux(t)\right) \\
				&=1-F(t+(1+2\varepsilon)\aux(t))-\left(1-F\left(t+\chi(\varepsilon)\aux(t)\right)\right),
\end{align*}
where $\chi(\varepsilon)=1/(\varepsilon(1+\varepsilon))$.
In what follows we show that the previous upper and lower bounds on (\ref{eq:lower-von}) imply the tail equivalence between $V$ and $F_{\phi}$.
By applying repeatedly Lemma~\ref{lem:von-mises-prop} and the same argument used in the proof of Lemma~\ref{lem:mises-to-gumbel} for $b_n = t$ and $t = \xi(\varepsilon), (1+2\varepsilon)$ and $\chi(\varepsilon)$, respectively, we have that the following holds,
 \begin{align*}
\lim_{t\to \infty}\frac{1-F\left(t+\xi(\varepsilon)\aux(t)\right)}{1-V(t)}&=\lim_{t\to \infty}\eta(t+\xi(\varepsilon)\aux(t))\cdot \frac{1-V\left(t+\xi(\varepsilon)\aux(t)\right)}{1-V(t)}\\
&=\eta^\star e^{-\xi(\varepsilon)},\\
\lim_{t\to \infty}\frac{1-F\left(t+(1+2\varepsilon)\aux(t)\right)}{1-V(t)}&=\lim_{t\to \infty}\eta(t+(1+2\varepsilon)\aux(t))\cdot \frac{1-V\left(t+(1+2\varepsilon)\aux(t)\right)}{1-V(t)}\\
&=\eta^\star e^{-(1+2\varepsilon)},\\
\lim_{t\to \infty}\frac{1-F\left(t+\chi(\varepsilon)\aux(t)\right)}{1-V(t)}&=\lim_{t\to \infty}\eta(t+\chi(\varepsilon)\aux(t))\cdot \frac{1-V\left(t+\chi(\varepsilon)\aux(t)\right)}{1-V(t)}\\
&=\eta^\star e^{-\chi(\varepsilon)}.
\end{align*}

\noindent Since the above holds for every $\varepsilon\in (0,1/3)$, and observing that the function $\xi(\varepsilon)$ is increasing in the interval $(0,1/3)$ we conclude that 
\begin{align*}
\limsup_{t\to \infty}\frac{\PP(\phi_F(v)>t)}{1-V(t)}&\le \eta^{\star} \inf_{\varepsilon\in (0,1/3)}\exp\left(-\frac{1-\varepsilon}{1-\varepsilon^2+\varepsilon}\right)=\frac{\eta^{\star}}{e},\\
\liminf_{t\to \infty}\frac{\PP(\phi_F(v)>t)}{1-V(t)}&\ge \eta^{\star} \inf_{\varepsilon\in (0,1/3)}\left\{\exp\left(-(1+2\varepsilon)\right)-\exp\left(-\frac{1}{\varepsilon(1+\varepsilon)}\right)\right\}=\frac{\eta^{\star}}{e},
\end{align*}
that in turn implies that the claimed limit exists and its value is equal to
\begin{equation*}
\lim_{t\to \infty}\frac{1-F_{\phi}(t)}{1-V(t)}=\lim_{t\to \infty}\frac{\PP(\phi_F(v)>t)}{1-V(t)}=\frac{\eta^{\star}}{e}.
\end{equation*}
This concludes that $V$ and $F_{\phi}$ are tail equivalent with an asymptotic tail ratio equal to $\eta/e$.
\end{proof}

\section{Competition Complexity Guarantees: Proof of Theorem \ref{thm:comp-complexity-evt}}\label{sec:analysis-competition}
In this section, we prove Theorem~\ref{thm:comp-complexity-evt} about the large market competition complexity of distributions satisfying the extreme value condition. 
We start by obtaining a very useful quantile characterization of the expected value of the optimal policy; the quantile depends only on $\gamma$ and $n$. Afterward, we use this characterization and jointly with a set of lemmas on the order statistics to show our main result. 
We start by providing a simple proposition about the optimal policy. Recall that $\{G_n(F)\}_{n\in \NN}$ satisfies the following recurrence: $G_0(F)=0$, $G_1(F)=\EE(X)$, and $G_{n+1}(F)=\EE(\max(X,G_n(F)))$, where $X$ is distributed according to $F$.
\begin{lemma}\label{lem:opt-dp}
For every $n \in \N$, and every distribution $F$, the following holds:
\begin{enumerate}[label=\normalfont(\roman*)]    
    \item $\Gmax{n+1} = \Gmax{n} + \int^{\omega_1(F)}_{\Gmax{n}} (1 - F(u))\dif u$.\label{lem:opt-dp-g}
    \item $\Gmax{n+1} = \Gmax{n} F(\Gmax{n}) + \int^{1 - F(\Gmax{n})}_0 {F^{\inv}(1-u)}\dif u$.\label{lem:opt-dp-inv}
\end{enumerate}    
\end{lemma}

\begin{proof}
Part \ref{lem:opt-dp-g} is obtained directly as
\begin{align*}
G_{n+1}(F)&=\EE(\max(X,G_n(F)))\\
&=G_n(F)+\EE(\max(0,G_n(F)-X))=\int^{\omega_1(F)}_{\Gmax{n}} {\prn{1 - F(u)} \dif u}.
\end{align*}
Part \ref{lem:opt-dp-inv} is obtained by taking the equality \ref{lem:opt-dp-g}, changing variables, and integrating by parts:
\begin{align*}
G_{n+1}(F) &= G_n(F) - G_n(F) \prn{1 - F(G_n(F))} + \int^{1 - F(G_n(F))}_0 {F^{-1}(1-u) \dif u} \\
&= G_n(F) F(G_n(F)) + \int^{1 - F(G_n(F))}_0 {F^{-1}(1-u) \dif u}.\qedhere
\end{align*}
\end{proof}

Our analysis relies heavily on the theory of \emph{regularly varying functions}, originally developed by Karamata. Regularly varying functions are, roughly speaking, functions that behave asymptotically like power functions. 
We briefly define them here; for more information on the topic, see e.g., \cite{libroextremo,bingham-rv,dehaan-ferreira-evt}.
\begin{definition}\label{def:reg-var}
Let $f: \Rp \to \Rp$ be a Lebesgue measurable function. 
We say that $f$ is regularly varying if, for some $\alpha \in \R$ and every $x > 0$, we have
$\lim_{t \to \infty} f(tx)/f(t) = x^\alpha.$
In this case, we indicate this as $f \in \rva{\alpha}$.
\end{definition}
In the above definition, $\alpha$ is called the \emph{index} of regular variation, and whenever $\alpha = 0$, we say that $f$ is \emph{slowly varying}. Furthermore, we say that $f(x)$ is regularly varying at $0$ if and only if $f(1/x)$ is regularly varying at infinity. If $f \in \rva{\alpha}$, then $L(x) = f(x)/x^\alpha \in \rva{0}$. 
One can represent any $f \in \rva{\alpha}$ as $f(x) = x^\alpha L(x)$, where $L$ is a slowly varying function.
The main tool from the theory of regular variation that we use to analyze the competition complexity is Karamata's Theorem (see, e.g., \cite[Theorem B.1.5]{dehaan-ferreira-evt}), which enables us to simplify the integrals in Lemma~\ref{lem:opt-dp}\ref{lem:opt-dp-inv} and obtain approximation recurrences for the optimal policy when $n$ is large.

We use the notation $a(n) \approx b(n)$ when $\lim_{n \to \infty} (a(n) - b(n)) = 0$. We also extend this notation for functions $f(x)\approx g(x)$ whenever $\lim_{x\to \infty}(f(x)-g(x))=0$. In what follows, for notation simplicity, we denote $\EE(\textstyle\max\{X_1,X_2,\ldots,X_n\})$ by $E_n(F)$. Recall that $\gamma$ partitions the space of all distributions with extreme value into the equivalence classes $D_{\gamma}$ with $\gamma \in \RR$: For $\gamma > 0$, we recover the Fr\'{e}chet family of distributions, for $\gamma = 0$, we get the Gumbel family, and for $\gamma < 0$, we obtain the Reverse Weibull family. We also say that a distribution $F$ is in $D_\gamma$ if it has extreme value index $\gamma$. In other words, $F$ is in $D_\gamma$ for $\gamma > 0$ implies that $F$ has extreme in the Fr\'{e}chet family with parameter $\alpha = 1/\gamma$, $F$ is in $D_\gamma$ for $\gamma = 0$ implies that $F$ has extreme Gumbel, and $F$ is in $D_\gamma$ for $\gamma < 0$ implies that $F$ has extreme in the Reverse Weibull family with parameter $\alpha = -1/\gamma$. Thus, from the extreme value theorem, it follows that $F \in D_\gamma$ implies that there exist sequences $\{a_n\}_{n\in \NN}, \{b_n\}_{n\in \NN}$ such that
\begin{equation}\label{eq:gnedenko-mu-approx-max}
\lim_{n \to \infty} F^n(a_n x + b_n) = G_\gamma(x)= \begin{cases}
\exp\prn{-\prn{1+\gamma x}^{-\f{1}{\gamma}}}, & \text{if } \gamma \neq 0, \\
\exp\prn{-\exp\prn{-x}}, & \text{if } \gamma = 0.
\end{cases}
\end{equation}

\begin{restatable}{lemma}{muApprox}\label{lem:mu-approx}
For $\gamma\in \RR$ and every $F\in D_{\gamma}$ the following holds:
\[
E_n(F)\approx \begin{cases}
\Gamma(1 - \gamma) F^{\inv}(1-1/n) & \text{for } \gamma \in (0,1), \\
F^{\inv}\prn{1 - \frac{\exp({-\gamma^\star})}{n}} & \text{for } \gamma = 0, \\
\omega_1(F) - \Gamma(1 - \gamma) \prn{\omega_1(F) - F^{\inv}(1-1/n)} & \text{for } \gamma < 0,
\end{cases}
\]
where $\gamma^\star \approx 0.577$ is the Euler-Mascheroni constant.
\end{restatable}
Before proving the lemma, we recall a useful fact from extreme value theory.
For $M_n^1 = \max\set{X_1, \dots, X_n}$ and $U(x)=F^{-1}(1-x)$, when \eqref{eq:gnedenko-mu-approx-max} is satisfied, we can take the following sequences (see, e.g., \cite[Corollary 1.2.4]{dehaan-ferreira-evt}):
\begin{enumerate}
    \item $a_n = U(1/n)$ and $b_n = 0$, if $\gamma \in (0,1)$,
    \item $b_n = U(1/n)$ and appropriately chosen $a_n$ if $\gamma = 0$, and
    \item $a_n = \omega_1(F) - U(1/n)$ and $b_n = \omega_1(F)$, if $\gamma < 0$,
\end{enumerate}
Let $Y_n = (M_n^1 - b_n)/a_n$. 
The above implies that the sequence $\{Y_n\}_{n\in \NN}$ converges in distribution to a random variable $Z$ distributed according to $G_\gamma$. Notice that for $\gamma\in (0,1)$ we have $\E[Z] = \Gamma(1 - \gamma)$, for $\gamma = 0$ we have $\E[Z] = \gamma^{\star}$, and for $\gamma < 0$ we have $\E[Z] = - \Gamma(1 - \gamma)$; see, e.g., \cite[Theorem 5.3.1]{dehaan-ferreira-evt}.
Now we have all the ingredients to prove Lemma \ref{lem:mu-approx}.
\begin{proof}[Proof of Lemma \ref{lem:mu-approx}]
For large enough $n$ and $\gamma \in (0,1)$ we have
$\E(Y_n) \approx \Gamma(1 - \gamma)$ if and only if $E_n(F) \approx \Gamma(1 - \gamma) a_n = \Gamma(1 - \gamma) F^{\inv}\prn{1 - {1}/{n}}.$
For $\gamma = 0$, we have $\E(Y_n) \approx \gamma^{\star}$ if and only if
\begin{equation}\label{eq:mu-approx-max-1}
E_n(F) \approx a_n \gamma^{\star} + b_n = a_n \gamma^{\star} + U\prn{{1}/{n}} = a_n \gamma^{\star} + F^{\inv}\prn{1 - {1}/{n}}.
\end{equation}
Also, for $\gamma = 0$, by \cite[Theorem 1.1.6]{dehaan-ferreira-evt}, we have that, for any $x > 1$
\begin{equation}\label{eq:mu-approx-max-2}
a_n \approx \frac{U\prn{{1}/{(x n)}} - U\prn{{1}/{n}}}{\log{x}} = \frac{F^{\inv}\prn{1 - {1}/{(x n)}} - F^{\inv}\prn{1 - {1}/{n}}}{\log{x}}.
\end{equation}
Combining \eqref{eq:mu-approx-max-1} and \eqref{eq:mu-approx-max-2}, we get
\[
E_n(F) \approx \frac{\gamma^{\star}}{\log{x}} F^{\inv}(1-1/(xn)) + \prn{1 - \frac{\gamma^{\star}}{\log{x}}} F^{\inv}(1-1/n),
\]
and setting $x = e^{\gamma^{\star}}$ yields
$E_n(F) \approx F^{\inv}\prn{1 - {e^{-\gamma^{\star}}}/{n}}.$
Finally, for $\gamma < 0$, we have $\E(Y_n) \approx -\Gamma(1 - \gamma)$ if and only if $E_n(F) \approx -\Gamma(1 - \gamma) a_n  + b_n = \omega_1(F) - \Gamma(1 - \gamma) \prn{\omega_1(F) - F^{\inv}(1-1/n)}.$
\end{proof}

\begin{restatable}{lemma}{multQApprox}\label{lem:mult-quantile-approx}
For $\gamma\in (-\infty,1)$, $F\in D_{\gamma}$, and $c > 0$, the following holds:
\begin{enumerate}[label=\normalfont(\roman*)]  
\item For $\gamma \in (0,1)$, we have $F^{\inv}\prn{1 - {c}/{n}} \approx c^{-\gamma} F^{\inv}\prn{1 - {1}/{n}}$.
\item For $\gamma \in (0,1)$, we have $F^{\inv}\prn{1 - \frac{1}{n-c}} \approx \prn{\frac{n}{n-c}}^{-\gamma} F^{\inv}\prn{1 - {1}/{n}}$.
\item For $\gamma < 0$, we have $\omega_1(F) - F^{\inv}\prn{1 - {c}/{n}} \approx c^{-\gamma} \prn{\omega_1(F) - F^{\inv}\prn{1 - {1}/{n}}}$.
\item For $\gamma < 0$, we have $\omega_1(F) - F^{\inv}\prn{1 - \frac{1}{n-c}} \approx \prn{\frac{n}{n-c}}^{-\gamma} \prn{\omega_1(F) - F^{\inv}\prn{1 - {1}/{n}}}$.
\end{enumerate}
\end{restatable}

\begin{proof}
(i). Let $\gamma \in (0,1)$. Since $F \in D_\gamma$, we have $U\circ \f{1}{x} \in \text{RV}_{\gamma}$; see, e.g., \cite[Corollary 1.2.10]{dehaan-ferreira-evt}. In particular, we have
\[
\lim_{n \to \infty} {\frac{F^{\inv}(1-1/(xn))}{F^{\inv}(1-1/n)}}=\lim_{n \to \infty} \frac{U({1}/{(x n)})}{U({1}/{n})} = x^{\gamma},
\]
for all $x > 0$. 
In particular, for $x = {1}/{c}$ we get
$F^{\inv}(1-c/n) \approx c^{-\gamma} F^{\inv}(1-1/n)$.

(ii). Next, by the Potter bounds on regularly varying functions (see, e.g., \cite[Theorem 1.5.6]{bingham-rv}), since $U\circ \f{1}{x} \in \text{RV}_{\gamma}$, we have that for every $\delta_1, \delta_2 > 0$ and every $x,y > 0$, there exists a $x_0 = x_0(\delta_1, \delta_2) > 0$ such that for all $x,y \geq x_0$,
\[
(1- \delta_1) \prn{\frac{x}{y}}^{\gamma - \delta_2} \leq \frac{U(1/x)}{U(1/y)} \leq (1+\delta_1) \prn{\frac{x}{y}}^{\gamma + \delta_2},
\]
which implies that
\[
\limsup_{\delta_1, \delta_2} (1- \delta_1) \prn{\frac{n-c}{n}}^{\gamma - \delta_2} \leq \lim_{n \to \infty} \frac{U\prn{\frac{1}{n-c}}}{U\prn{\frac{1}{n}}} \leq \liminf_{\delta_1, \delta_2} (1+\delta_1) \prn{\frac{n-c}{n}}^{\gamma + \delta_2},
\]
and thus
\[
\lim_{n \to \infty} \frac{U\prn{\frac{1}{n-c}}}{U\prn{\frac{1}{n}}} = \prn{\frac{n}{n-c}}^{-\gamma},
\]
yielding $F^{\inv}\prn{1 - \frac{1}{n-c}} \approx \prn{\frac{n}{n-c}}^{-\gamma} F^{\inv}\prn{1 - {1}/{n}}$.

(iii). Next, let $\gamma < 0$. 
Since $F \in D_\gamma$, we have $\omega_1(F) - U\circ \f{1}{x}\in \text{RV}_{\gamma}$: see, e.g., \cite[Corollary 1.2.10]{dehaan-ferreira-evt}. In particular, we have
\[
\lim_{n \to \infty} {\frac{\omega_1(F) - F^{\inv}(1-1/(xn))}{\omega_1(F) - F^{\inv}(1-1/n)}}=\lim_{n \to \infty} \frac{\omega_1(F) - U(1/(xn))}{\omega_1(F) - U(1/n)} =x^{\gamma},
\]
for all $x > 0$. 
In particular, for $x = {1}/{c}$ we get
$x^{\star} - F^{\inv}(1-c/n) \approx c^{-\gamma} \prn{x^{\star} - F^{\inv}(1-1/n)}$.

(iv). Finally, by the Potter bounds on regularly varying functions (see, e.g., \cite[Theorem 1.5.6]{bingham-rv}), since $\omega_1(F) - U\circ \f{1}{x} \in \text{RV}_{\gamma}$, we have that for every $\delta_1, \delta_2 > 0$ and every $x,y > 0$, there exists a $x_0 = x_0(\delta_1, \delta_2) > 0$ such that for all $x,y \geq x_0$,
\[
(1- \delta_1) \prn{\frac{x}{y}}^{\gamma - \delta_2} \leq \frac{\omega_1(F) - U(1/x)}{\omega_1(F) - U(1/y)} \leq (1+\delta_1) \prn{\frac{x}{y}}^{\gamma + \delta_2},
\]
which implies that
\[
\limsup_{\delta_1, \delta_2} (1- \delta_1) \prn{\frac{n-c}{n}}^{\gamma - \delta_2} \leq \lim_{n \to \infty} \frac{\omega_1(F) - U\prn{\frac{1}{n-c}}}{\omega_1(F) - U\prn{\frac{1}{n}}} \leq \liminf_{\delta_1, \delta_2} (1+\delta_1) \prn{\frac{n-c}{n}}^{\gamma + \delta_2},
\]
and thus
\[
\lim_{n \to \infty} \frac{\omega_1(F) - U\prn{\frac{1}{n-c}}}{\omega_1(F) - U\prn{\frac{1}{n}}} = \prn{\frac{n}{n-c}}^{-\gamma},
\]
yielding $\omega_1(F) - F^{\inv}\prn{1 - \frac{1}{n-c}} \approx \prn{\frac{n}{n-c}}^{-\gamma} \prn{\omega_1(F) - F^{\inv}\prn{1 - {1}/{n}}}$.
\end{proof}

Next, we use Lemmas~\ref{lem:mu-approx} and~\ref{lem:mult-quantile-approx} to characterize the ratio of the prophet's expected value for $n-1$ and $n$ as $n \to \infty$. 
\begin{restatable}{lemma}{muRatioApprox}\label{lem:mu-ratio-approx}
For $\gamma\in (-\infty,1)$ and $F\in D_{\gamma}$, the following holds:
\begin{enumerate}[label=\normalfont(\roman*)]  
\item For $\gamma \in (0,1)$, we have $E_{n-1}(F)/E_{n}(F) = 1 - \gamma/n + o({1}/{n})$.
\item For $\gamma < 0$, we have ${(\omega_1(F) - E_{n-1}(F))}/{(\omega_1(F) - E_{n}(F))} = 1 - \gamma/n + o(1/n)$.
\end{enumerate}
\end{restatable}

\begin{proof}
Let $\gamma \in (0,1)$. 
By Lemma~\ref{lem:mu-approx} we have $E_{n-1}(F) \approx \Gamma(1 - \gamma) F^{\inv}(1-1/(n-1))$, and using Lemma~\ref{lem:mult-quantile-approx}, we get
\[
E_{n-1}(F) \approx \Gamma(1 - \gamma) \prn{1+\frac{1}{n-1}}^{-\gamma} F^{\inv}(1-1/n).
\]
Again by Lemma~\ref{lem:mu-approx}, we have
$E_{n}(F) \approx \Gamma(1 - \gamma) F^{\inv}(1-1/n)$
and thus
\[
\frac{E_{n-1}(F)}{E_{n}(F)} \approx \prn{1+\frac{1}{n-1}}^{-\gamma} = 1 - {\gamma}/{n} + o(1/n).
\]
Similarly, for $\gamma < 0$, by Lemma~\ref{lem:mu-approx} we have $\omega_1(F) - E_{n-1}(F) \approx \Gamma(1 - \gamma) F^{\inv}(1-1/(n-1))$ and using Lemma~\ref{lem:mult-quantile-approx}, we get
\[
x^{\star} - E_{n-1}(F) \approx \Gamma(1 - \gamma) \prn{1 +\frac{1}{n-1}}^{-\gamma} F^{\inv}(1-1/n).
\]
By Lemma~\ref{lem:mu-approx}, we have
$\omega_1(F) - E_{n}(F) \approx \Gamma(1 - \gamma) F^{\inv}(1-1/n),$
and thus
\[
\frac{\omega_1(F) - E_{n-1}(F)}{\omega_1(F) - E_{n}(F)} \approx \prn{1+\frac{1}{n-1}}^{-\gamma} = 1 - {\gamma}/{n} + o(1/n).\qedhere
\]
\end{proof}

\begin{lemma}\label{lem:gn-done}
When $\gamma<1$, for every $F\in D_\gamma$ we have
$\Gmax{n+1} \approx \Gmax{n} \prn{1 + \frac{\gamma}{1 - \gamma} \prn{1 - F(\Gmax{n})}}.$
\end{lemma}

\begin{proof}
By Karamata's theorem (see, e.g., \cite[Theorem B.1.5]{dehaan-ferreira-evt}), when $\gamma <1$ we have
$\int^x_0 {F^{\inv}(1-u) \dif u} \approx {x F^{\inv}(1-x)}/{(1 - \gamma)}.$
In particular, this implies that
\begin{align*}
\int^{1 - F(\Gmax{n})}_0 {F^{\inv}(1 -u) \dif u} &\approx \frac{\prn{1 - F(\Gmax{n})}F^{\inv}\prn{1 - (1 - F(\Gmax{n}))}}{1-\gamma}\\
&= \frac{\Gmax{n}\prn{1 - F(\Gmax{n})}}{1-\gamma}.
\end{align*}
Thus, by Lemma~\ref{lem:opt-dp}\ref{lem:opt-dp-inv},
\begin{align*}
\Gmax{n+1} &\approx \Gmax{n} \prn{F(\Gmax{n}) + \frac{1 - F(\Gmax{n})}{1-\gamma}}= \Gmax{n} \prn{1 + \frac{\gamma}{1 - \gamma} \: \: \prn{1 - F(\Gmax{n})}}.\qedhere
\end{align*}
\end{proof}
In the following lemma, we provide an explicit approximation of $\Gmax{n}$ with respect to the quantile function of the distribution and the value $\gamma$.
\begin{lemma}\label{lem:alg-quantile}
When $\gamma < 1$, for every $F \in D_\gamma$ we have
$\Gmax{n} \approx F^{\inv}\prn{1 - {(1 - \gamma)}/{(n+1)}}.$
\end{lemma}

\begin{proof}
From \cite{KK91}, we know that the competitive ratio of the prophet inequality for distributions with extreme value is equal to a non-zero constant, which implies that
${\Gmax{n+1}}/{E_{n+1}(F)} - {\Gmax{n}}/{E_{n}(F)} \approx 0.$
Thus, for $\gamma \in (0,1)$, we obtain
\begin{align*}
& \frac{\Gmax{n+1}}{E_{n+1}(F)} - \frac{\Gmax{n}}{E_{n}(F)} \approx 0 \\
&\iff \frac{\Gmax{n}}{E_{n}(F)} \prn{\frac{E_{n}(F)}{E_{n+1}(F)} \prn{1 + \frac{\gamma}{1 - \gamma} \prn{1 - F(\Gmax{n})}} - 1} \approx 0 \\
&\iff \prn{1 - \frac{\gamma}{n+1} + o\left(\frac{1}{n+1}\right)} \prn{1 + \frac{\gamma}{1 - \gamma} \prn{1 - F(\Gmax{n})}} - 1 \approx 0 \\
&\iff \prn{1 - F(\Gmax{n})} \frac{\gamma}{1 - \gamma} \approx \frac{\gamma}{n+1} \\
&\iff 1 - F(\Gmax{n}) \approx \frac{1 - \gamma}{n+1}\iff \Gmax{n} \approx F^{\inv}\prn{1 - \frac{1 - \gamma}{n+1}},
\end{align*}
where the first equivalence follows from Lemma~\ref{lem:gn-done}, the second equivalence follows from Lemma~\ref{lem:mu-ratio-approx} and the third equivalence follows by ignoring lower-order terms.
For $\gamma = 0$, we need the following claim.
\begin{claim}\label{clm:g-approx-for-zero}
For $F \in D_0$ we have
$\int^{1 - F(\Gmax{n})}_0 {F^{\inv}\prn{1-u} \dif u} \approx F^{\inv}\prn{1 - \frac{1}{n+1}} \prn{1 - F(\Gmax{n})}.$
\end{claim}
\begin{proof}
The claim follows directly from \cite[Corollary 1.2.15]{dehaan-ferreira-evt} for $t = e/{(1 - F(\Gmax{n}))}$.
\end{proof}
Combining Claim~\ref{clm:g-approx-for-zero} with Lemma~\ref{lem:opt-dp}\ref{lem:opt-dp-inv} and Lemma~\ref{lem:gn-done}, we get
\[
\Gmax{n}\prn{1 - F(\Gmax{n})} \approx F^{\inv}\prn{1 - \frac{1}{n+1}} \prn{1 - F(\Gmax{n})} \iff \Gmax{n} \approx F^{\inv}\prn{1 - \frac{1}{n+1}}.
\]
When $\gamma<0$, we necessarily have $\omega_1(F) < \infty$. 
The analysis follows similarly to the $\gamma \in (0,1)$ case, but for the value $\omega_1(F) - \Gmax{n}$ instead of $\Gmax{n}$.
Notice that, by Lemma~\ref{lem:opt-dp}\ref{lem:opt-dp-g} we have
\begin{align*}
\omega_1(F) - \Gmax{n+1} &= \int^{\omega_1(F)}_{\Gmax{n}} {F(u) \dif u}.
\end{align*}
\begin{claim}\label{lem:max-neg-g-int}
For $\gamma < 0$  we have
$\int^{\omega_1(F)}_{\Gmax{n}} {F(u) \dif u} \approx \prn{1 + \frac{\gamma}{1-\gamma} \prn{1 - F(\Gmax{n})} } \prn{\omega_1(F) - \Gmax{n}}.$
\end{claim}
\begin{proof}
Let $y = 1 - F(u)$, which implies that $\dif u = \prn{F^{\inv}\prn{1-y}}' \dif y$. Then,
\begin{align*}
\int^{\omega_1(F)}_{\Gmax{n}} {F(u) \dif u} 
&= \omega_1(F) - \Gmax{n} F(\Gmax{n}) - \int^{1 - F(\Gmax{n})}_0 {F^{\inv}\prn{1-y} \dif y} \\
&= F(\Gmax{n}) \prn{\omega_1(F) - \Gmax{n}} + \int^{1 - F(\Gmax{n})}_0 {\prn{\omega_1(F) - F^{\inv}\prn{1-y}} \dif y}.
\end{align*}
Next, let $z = {1}/{y}$, which implies $\dif y = -z^{-2} \dif z$, and thus
\begin{align*}
\int^{\omega_1(F)}_{\Gmax{n}} {F(u) \dif u} &=  F(\Gmax{n}) \prn{\omega_1(F) - \Gmax{n}} + \int^\infty_{{1}/{(1 - F(\Gmax{n}))}} {z^{-2} \prn{\omega_1(F) - F^{\inv}\prn{1-{1}/{z}}} \dif z}.
\end{align*}
For $\gamma < 0$, the function $g(z) = \omega_1(F) - F^{\inv}\prn{1-{1}/{z}} \in \text{RV}_\gamma$ by \cite[Corollary 1.2.10]{dehaan-ferreira-evt} and $1+\gamma < 1$ which, by the general form of Karamata's theorem \cite[Theorem 1.5.11]{bingham-rv} for $\sigma = -2$ implies that
\[
\int^\infty_{{1}/{(1 - F(\Gmax{n}))}} {z^{-2} \prn{\omega_1(F) - F^{\inv}\prn{1 - {1}/{z}}} \dif z} \approx \frac{1 - F(\Gmax{n})}{1 - \gamma} \prn{\omega_1(F) - \Gmax{n}}.
\]
Therefore,
\begin{align*}
\int^{\omega_1(F)}_{\Gmax{n}} {F(u) \dif u} &\approx \prn{F(\Gmax{n}) + \frac{1 - F(\Gmax{n})}{1 - \gamma}} \prn{\omega_1(F) - \Gmax{n}} \\
&= \prn{1 - \prn{1 - F(\Gmax{n})} + \frac{1 - F(\Gmax{n})}{1 - \gamma}} \prn{\omega_1(F) - \Gmax{n}} \\
&= \prn{1 - \prn{1 - F(\Gmax{n})} \prn{1 - \frac{1}{1 - \gamma}}} \prn{\omega_1(F) - \Gmax{n}} \\
&= \prn{1 + \frac{\gamma}{1-\gamma} \prn{1 - F(\Gmax{n})} } \prn{\omega_1(F) - \Gmax{n}}.\qedhere
\end{align*}
\end{proof}
Thus,
\[
\frac{\omega_1(F) - \Gmax{n+1}}{\omega_1(F) - \Gmax{n}} \approx 1 + \frac{\gamma}{1-\gamma} \prn{1 - F(\Gmax{n})}.
\]
However, notice that
\begin{align*}
\frac{\omega_1(F) - \Gmax{n+1}}{\omega_1(F) - \Gmax{n}} &= \frac{\omega_1(F) - \Gmax{n+1}}{\omega_1(F) - E_{n+1}(F)} \cdot \frac{\omega_1(F) - E_{n}(F)}{\omega_1(F) - \Gmax{n}} \cdot \frac{\omega_1(F) - E_{n+1}(F)}{\omega_1(F) - E_{n}(F)}\\
&\approx \frac{\omega_1(F) - E_{n+1}(F)}{\omega_1(F) - E_{n}(F)},
\end{align*}
where the last asymptotic equality follows from the fact that the competitive ratio of the prophet inequality converges asymptotically to a non-zero constant \cite{KK91}. By Lemma~\ref{lem:mu-ratio-approx}, we have
$${(\omega_1(F) - E_{n+1}(F))}/{(\omega_1(F) - E_{n}(F))} \approx 1 - {\gamma}/{(n+1)} + o(1/n),$$
and thus, ignoring lower-order terms, we obtain
\[
1 + \frac{\gamma}{1-\gamma} \prn{1 - F(\Gmax{n})} \approx 1 - \frac{\gamma}{n+1},
\]
which implies that $\Gmax{n} \approx F^{\inv}\prn{1 - \frac{1-\gamma}{n+1}}$.
\end{proof}
We are now ready to prove Theorem~\ref{thm:comp-complexity-evt}.
\begin{proof}[Proof of Theorem \ref{thm:comp-complexity-evt}]
First, let $\gamma > 0$. For every $c>0$ we have
\[
\Gmax{cn} \approx F^{\inv}\prn{1 - \frac{1-\gamma}{c \: n+1}} \approx F^{\inv}\prn{1 - \frac{1-\gamma}{c \: n}} \approx \prn{\frac{1-\gamma}{c}}^{-\gamma} F^{\inv}(1-1/n),
\]
where the first asymptotic equality follows from Lemma~\ref{lem:alg-quantile}, and the second one follows from Lemma~\ref{lem:mult-quantile-approx}. 
From Lemma~\ref{lem:mu-approx} we have
$E_{n}(F) \approx \Gamma(1-\gamma) F^{\inv}(1-1/n)$and observe that
${((1-\gamma)/{c})}^{-\gamma} \geq \Gamma(1-\gamma)$ if and only if $c \geq (1-\gamma) \prn{\Gamma(1-\gamma)}^{\f{1}{\gamma}}$ and thus,
$$\ACCMmax = \inf\set{c>0 :\prn{\frac{1-\gamma}{c}}^{-\gamma} \geq \Gamma(1-\gamma)} = (1-\gamma) \prn{\Gamma(1-\gamma)}^{\f{1}{\gamma}}.$$
Next, let $\gamma = 0$. We have
$\Gmax{cn} \approx F^{\inv}\prn{1 - {1}/{(cn+1)}} \approx F^{\inv}\prn{1 - {1}/{(cn)}}$,
where the first asymptotic equality follows from Lemma~\ref{lem:alg-quantile}. 
From Lemma~\ref{lem:mu-approx} we have
$E_{n}(F) \approx F^{\inv}\prn{1 - {e^{-\gamma^\star}}/{n}},$
where $\gamma^\star \approx 0.577$ is the Euler-Mascheroni constant. Notice that since $F$ is a monotonically increasing function, we have
$F^{\inv}\prn{1 - {1}/{(cn)}} \geq F^{\inv}\prn{1 - {e^{-\gamma^\star}}/{n}}$ if and only if $c \geq e^{\gamma^{\star}}$.
Thus,
$\ACCMmax = e^{\gamma^{\star}}.$
Finally, let $\gamma < 0$ and recall that, in this case, $\omega_1(F) < \infty$.
We have
$\Gmax{cn} \approx F^{\inv}\prn{1 - {(1-\gamma)}/{(cn+1)}} \approx F^{\inv}\prn{1 - {(1-\gamma)}/{(cn)}}.$
From Lemma~\ref{lem:mu-approx}, we know that
$E_{n}(F) \approx \omega_1(F) - \Gamma(1-\gamma) \prn{\omega_1(F) - F^{\inv}(1-1/n)}$, and therefore, by Lemma~\ref{lem:alg-quantile}, $\Gmax{cn} \geq E_{n}(F)$ if and only if
\begin{align*}
\prn{\frac{1-\gamma}{c}}^{-\gamma} \prn{\omega_1(F) -  F^{\inv}(1-1/n)} &\leq \Gamma(1-\gamma) \prn{\omega_1(F) - F^{\inv}(1-1/n)}.
\end{align*} 
We conclude that
$\ACCMmax = \inf\set{c>0 : \prn{\frac{1-\gamma}{c}}^{-\gamma} \geq \Gamma(1-\gamma)} = (1-\gamma) \prn{\Gamma(1-\gamma)}^{\f{1}{\gamma}}.$
\end{proof}

\end{document}